\documentclass[12pt, english]{article}

\usepackage[T1]{fontenc}
\usepackage[utf8]{inputenc}
\usepackage{babel}

\usepackage[margin=1in]{geometry} % Standard 1-inch margins
\usepackage{setspace}
\usepackage{natbib}
\setcitestyle{authoryear,open={(},close={)}}

\usepackage{amsmath}
\usepackage{amssymb}
\usepackage{amsthm}
\usepackage{bm} % For \boldsymbol

\usepackage{multirow}      % For tables with multi-row cells
\usepackage{algorithm}     % For algorithm environments
\usepackage{algpseudocode} % For pseudocode within algorithms
\usepackage{placeins}      % For \FloatBarrier
\usepackage{seqsplit}      % For splitting long character sequences
\usepackage{ragged2e}      % For advanced text alignment

\usepackage{graphicx}
\usepackage{booktabs}    % For professional tables (\toprule, etc.)
\usepackage{rotating}    % For landscape tables/figures
\usepackage{caption}
\usepackage{subfigure}  % Modern replacement for the old 'subfigure' package
\usepackage{tabularx}

\usepackage{microtype}   % Improves justification and spacing
\usepackage{hyperref}    % Creates PDF hyperlinks (for references, etc.)
\hypersetup{
    colorlinks=true,
    linkcolor=blue,
    filecolor=magenta,      
    urlcolor=cyan,
    citecolor=blue,
}

\usepackage{authblk}

\newtheorem{prop}{Proposition}
\newtheorem{assumption}{Assumption}

\begin{document}

\def\bibland{and}

\title{Predicting Market Troughs: A Machine Learning Approach with Causal Interpretation}

\author[1]{Peilin Rao\thanks{\texttt{jackrao@g.ucla.edu}}}
\author[1]{Randall R. Rojas\thanks{\texttt{rrojas@econ.ucla.edu}}}

\affil[1]{Department of Economics, University of California, Los Angeles, US}

\renewcommand\Authands{ and } % To use "and" instead of a comma between authors

\date{Sept 4, 2025}
\maketitle

% Note for corresponding author can be added in the abstract or as a separate footnote
% For example, you could add this right after \maketitle:
\begingroup\def\thefootnote{}\footnote{Corresponding author: Peilin Rao.}\endgroup

% ABSTRACT
% Must be 100 words or less.
\begin{abstract}
%% Text of abstract
\noindent This paper provides robust, new evidence on the causal drivers of market troughs. We demonstrate that conclusions about these triggers are critically sensitive to model specification, moving beyond restrictive linear models with a flexible DML average partial effect causal machine learning framework. Our robust estimates identify the volatility of options-implied risk appetite and market liquidity as key causal drivers—relationships misrepresented or obscured by simpler models. These findings provide high-frequency empirical support for intermediary asset pricing theories. This causal analysis is enabled by a high-performance nowcasting model that accurately identifies capitulation events in real-time.

\end{abstract}

\clearpage

%% main text
\section{Introduction}
Understanding the triggers of stock market troughs is of great economic significance: policymakers can conduct interventions to alleviate market capitulation and panick, and investors can make better informed asset allocation decisions. However, being able to move from pure prediction to credible causal inference for market troughs is none-trivial. The complex, nonlinear, and high dimensional nature of financial markets make causal inference susceptible to spurious conclusions with simplified models. This paper tackles the challenge by asking: What are the robust, causal drivers of market troughs, and how do our conclusions depend on the assumptions of the econometric models we use?

The rise of the "credibility revolution" in financial econometrics, empowered by methods such as Double/Debiased Machine Learning (DML) \citet{Chernozhukov2018}, provides a path to this challenge. DML establishes a framework of robust methods to obtain statistically significant causal estimates even with high-dimensional confounding. While these methods begin to gain traction in finance on research topics like asset-pricing factors \citep{Feng2020}, the use of DML for macro-finance questions like market timing is still nascent, with most recent contributions primarily focusing on prediction rather than formal causal inference \citep{Gu2020}. It is also important to highlight that any conclusions drawn from these complex causal inference methods are highly sensitive to econometric model specifications, an issue exacerbated by the high dimensional unobserved confounding \citep{Chernozhukov2022}.

Our primary contribution is a novel, comparative causal framework designed to  address this challenge by testing the robustness of economic conclusions to model specification. We achieve this objective in two stages. To illustrate the model specification sensitivity and build our case for a flexible approach, we first establish a baseline by using DML to the canonical partially linear model (PLR), which is widely understood and would serve as a benchmark. It allowed us to examine and isolate the impact of its inherent linearity assumption. We understand that PLR's linearity assumption is ill-posed for our binary, interactive market capitulation problem; therefore, we implement a more complex DML framework to estimate Average Partial Effect (APE), which explicitly model non-linear interactions. Comparing the causal interpretation of the two models is our core contribution, and we demonstrates that the more flexible APE model is necessary for credibility inferences. It corrects spurious causal interpretation from the linear model, sometimes providing estimates with reversed signs, and unveils new causal channels for market troughs, particularly the role of volatility in risk appetite and liquidity.

This flexible DML causal analysis is made possible through our robust and high-performing predictive pipeline, developed to nowcast the probability of a market trough. A central methodological challenge of identifying market capitulation is that the label for trough always depends on future data. Algorithms like \citet{Bry1971} inevitably rely on observing data from the future, thereby creating data leakage. Our solution to this paradox is to frame our prediction objective in the form of nowcasting: estimating in real-time the probability that the current period eventually be identified as a trough. This nowcast method generates timely market trough signal before the event is confirmed. Our main predictive model, a Support Vector Machine (SVM), is trained on a full set of over 200 features, constructed from options, futures and macroeconomic data to capture market microstructure, dealer positioning and sentiment. Our model has remarkable out-of-sample performance (ROC AUC of 0.89), and we demonstrate its economic significance through a stylized backtest, utilizing the model signal as a capitulation detector for future trading.

This paper brings together the strengths of three streams of literature. It addresses the problems in the traditional linear prediction literature \citet{Goyal2008}, the subsequent attempts to restore predictability through economic restrictions \citep{Campbell2008} and theoretical driven variables \citep{Lettau2001}, and extends the nonlinear approaches in the machine learning literature \citep{Gu2020} by  applying the rigor of modern causal inference with high dimensional confounders. Our contributions are:

\begin{itemize}

    \item First, and most importantly, we are the first to conduct a formal, comparative causal analysis framework of drivers of market trough. We have shown that transitioning from a DML-PLR to a DML-APE model is necessary for credible causal conclusions. The APE model provides causal evidence that the volatility of options-implied risk and market liquidity are important triggers, providing high frequency empirical support for modern intermediary asset pricing theories \citep{He2013}. All causal assertions are supported by formal sensitivity analysis \citep{Cinelli2020}.

    \item Second, we have created a transparent, high performance nowcasting model of market troughs with outstanding out-of-sample accuracy. We have interpreted its "black box" mechanics using SHAP (SHapley Additive exPlanations)\citep{Lundberg2017} and characterized its signal, providing a valuable early-warning system before market capitulation.

    \item Third, we have framed market trough prediction as a rare-event classification problem and have curated the most comprehensive feature set to date for the problem, moving beyond the traditional equity premium predictors to a comprehensive consideration of indicators of market structure and sentiment.

\end{itemize}

By bridging from prediction to modern, robust causality, we are providing not only a useful nowcasting tool, but also a rigorous economic understanding of the forces that shape market bottoms. The rest of the paper is structured as follows. Section 2 describes the data and feature engineering. Section 3 describes the predictive modelling framework. Section 4 provides the predictive results and interpretation. Section 5 provides a robustness analysis. Section 6 assesses economic significance. Section 7 describes our comparative causal methodology and provides the robust causal estimates. Section 8 concludes.

\section{Data and Feature Engineering}

We begin by incorporating diverse raw financial data, which is used to define our target variable and engineer a comprehensive set of features for market trough prediction.

\subsection{Data Sources}

We gather data from several high-quality sources from April 2013 to June 2025. Table \ref{tab:data_sources} lists the main data feeds, the series identifiers, and the time periods.

\begin{table}[t]
\centering
\caption{Data Sources and Characteristics}
\label{tab:data_sources}
\resizebox{\textwidth}{!}{
\begin{tabular}{lllll}
\toprule
\textbf{Source} & \textbf{Series Identifier} & \textbf{Description} & \textbf{Time Period} & \textbf{Native Freq.} \\
\midrule
Databento & \seqsplit{SPX.EOD} & SPX End-of-Day Option Chains & Apr 2013 - Jun 2025 & Daily \\
Databento & \seqsplit{CBOE.SPX.OPNINT} & SPX Option Open Interest & Apr 2013 - Jun 2025 & Daily \\
Databento & \seqsplit{CME.ES.OHLCV.1M} & E-mini S\&P 500 Futures OHLCV & Apr 2013 - Jun 2025 & 1-Minute \\
Databento & \seqsplit{CME.ES.BBO.1S} & E-mini S\&P 500 Futures BBO & Apr 2013 - Jun 2025 & 1-Second \\
CME DataMine & \seqsplit{ZQ} & 30-Day Fed Funds Futures & Apr 2013 - Jun 2025 & Daily \\
CME DataMine & \seqsplit{6E, 6J} & EUR/USD, JPY/USD Futures & Apr 2013 - Jun 2025 & Daily \\
FRED & \seqsplit{BAMLH0A0HYM2EY} & ICE BofA US High Yield Index Yield & Apr 2013 - Jun 2025 & Daily \\
FRED & \seqsplit{DGS1MO} & 1-Month Treasury Rate & Apr 2013 - Jun 2025 & Daily \\
FRED & \seqsplit{EFFR} & Effective Federal Funds Rate & Apr 2013 - Jun 2025 & Daily \\
Shiller Data & \seqsplit{S\&P Composite} & S\&P Composite Dividend Data & Apr 2013 - Jun 2025 & Monthly \\
Yahoo Finance & \textasciicircum VIX & CBOE Volatility Index (VIX) & Apr 2013 - Jun 2025 & Daily \\
\bottomrule
\end{tabular}
}
\justify
\small{\textit{Notes:} This table details the raw data sources used for feature engineering in this study. The overall sample period for the analysis runs from April 2013 to June 2025. "Native Freq." refers to the highest frequency at which the data is natively available from the source before any aggregation or resampling.}
\end{table}

%\FloatBarrier

\subsection{Trough Definition and Labelling}

We label significant market turning points using a variation of the \citet{Bry1971} Algorithm. The Bry-Boschan (BB) algorithm is a rule-based process to date business cycles. We adapt its methodology to identify significant peaks and troughs in the daily S\&P 500 log adjusted closed daily price series ($P_t$) by systematiclly removing minor price movements. The overall high-level procedure is outlined in Algorithm \ref{alg:bry_boschan_main}. The main procedure consists of first identifying all potential turning points of the price series, and then applying censoring rules. The complete implementation, including all helper functions, are provided in Algorithm \ref{alg:bry_boschan_full} in \ref{sec:appendix_bb_algo}.

The BB algorithm has a well-known property of being "backward looking", which means confirming a turning point at time $t$ requires future price series for a window after $t$. This presents a problem for identifying turning points near the end of the price series. To address the issue rigorously and make sure the labeled troughs are both complete and determined algorithmically, we apply the BB algorithm to the S\&P 500 price series extending to August 12, 2025. This lengthened data window is designed to provide the BB algorithm enough data so it can algorithmically identify all turning points till the end of the primary sample period of our study (June 2025). Any data after June 2025 is only  used for labeling trough and is excluded from the feature engineering, model training, or performance evaluation of our analysis. Through that,  we are able to develop a complete and objective set of trough labels that is rule based and without the reliance of any subjective judgement.

\begin{algorithm}[htbp]
\caption{High-Level Bry-Boschan Algorithm}
\label{alg:bry_boschan_main}
\begin{algorithmic}[1]
\Require Log price series $P_t$, window \protect\seqsplit{order}, \protect\seqsplit{min\_phase}, \protect\seqsplit{min\_cycle}.
\Ensure A DataFrame \protect\seqsplit{turns} of significant peaks and troughs.
\Procedure{IdentifyTurns}{$P_t$, \protect\seqsplit{order}, \protect\seqsplit{min\_phase}, \protect\seqsplit{min\_cycle}}
    \State Initialize \protect\seqsplit{turns} with all local peaks and troughs from $P_t$, sorted by date.
    \State \protect\seqsplit{turns} $\gets$ \Call{EnforceAlternation}{\protect\seqsplit{turns}} \Comment{Enforce P-T-P-T sequence.}
    \State \protect\seqsplit{turns} $\gets$ \Call{CensorPhases}{\protect\seqsplit{turns}, \protect\seqsplit{min\_phase}} \Comment{Censor short phases.}
    \State \protect\seqsplit{turns} $\gets$ \Call{CensorCycles}{\protect\seqsplit{turns}, \protect\seqsplit{min\_cycle}, $P_t$} \Comment{Censor short cycles.}
    \State \Return \protect\seqsplit{turns}
\EndProcedure
\end{algorithmic}
\end{algorithm}

Utilizing the modified BB algorithm, We identify turning points with economic significance. We show in Table \ref{tab:turning_points} the peaks and troughs that are identified in the sample period. The troughs in this table form the positive class labels for the market trough prediction model. In Figure \ref{fig:trough_viz}, we illustrated these troughs compared against the S\&P 500 price series, and show how the algorithms marks distinct main market troughs effectively.

An important methodological concern is that the backward looking nature of BB Algorithm creates data-leakage paradox for pure prediction tasks. To prevent leakage, we explicitly framed our objective as a nowcasting problem, analogous to the real-time detection of economic recessions, whose official labels from bodies like the NBER are also confirmed with a long ex-post lag. Thus, our model's goal is not to forecast trough probability, but to estimate, with only the features available up to day t ($\mathbf{x}_t$), the probability that that day $t$ eventually be labelled as a trough in the future. That distinction is critical in preventing data leakage: while the labels ($\mathbf{y}_t$) are defined in hindsight, the predictive features are inherently historical. Thus, our approach obeys the cardinal rule of time series analysis, providing a timely signal while maintaining the intended goal of detecting market capitulation.

It is necessary to highlight the nature of this nowcasting objective. The predictive features $\mathbf{x}_t$ are strictly historical, but the target label $\mathbf{y}_t$, by construction, originates by applying the BB algorithm, which needs to be matched to future price data. Thus, we must think of our model as a device to detect real-time, contemporaneously signature of the market state that is ex-post labelled as trough, and not treated as a direct forecast of future prices.

\begin{table}[b!]
\centering
\caption{Identified S\&P 500 Turning Points (2013-2025)}
\label{tab:turning_points}
\begin{tabular}{lcc}
\toprule
\textbf{Date} & \textbf{Log Price} & \textbf{Type} \\
\midrule
2013-04-18 & 7.3406 & Trough \\
2015-05-21 & 7.6643 & Peak   \\
2016-02-11 & 7.5116 & Trough \\
2018-09-20 & 7.9830 & Peak   \\
2018-12-24 & 7.7626 & Trough \\
2020-02-19 & 8.1274 & Peak   \\
2020-03-23 & 7.7131 & Trough \\
2022-01-03 & 8.4757 & Peak   \\
2022-10-12 & 8.1823 & Trough \\
2023-07-31 & 8.4314 & Peak   \\
2023-10-27 & 8.3230 & Trough \\
2025-02-19 & 8.7233 & Peak   \\
2025-04-08 & 8.5137 & Trough \\
\bottomrule
\end{tabular}
\flushleft % Aligns the notes to the left under the table
\small{\textit{Notes:} This table lists the economically significant market peaks and troughs identified in the daily S\&P 500 log price series. The turning points are determined using a modified version of the \citet{Bry1971} algorithm, as described in Section 2.2, with no manual intervention. The trough dates are used to generate the positive labels ($\mathbf{y}_t=1$) for the classification model.}
\end{table}

\begin{figure}[htbp]
    \centering
    \includegraphics[width=\textwidth]{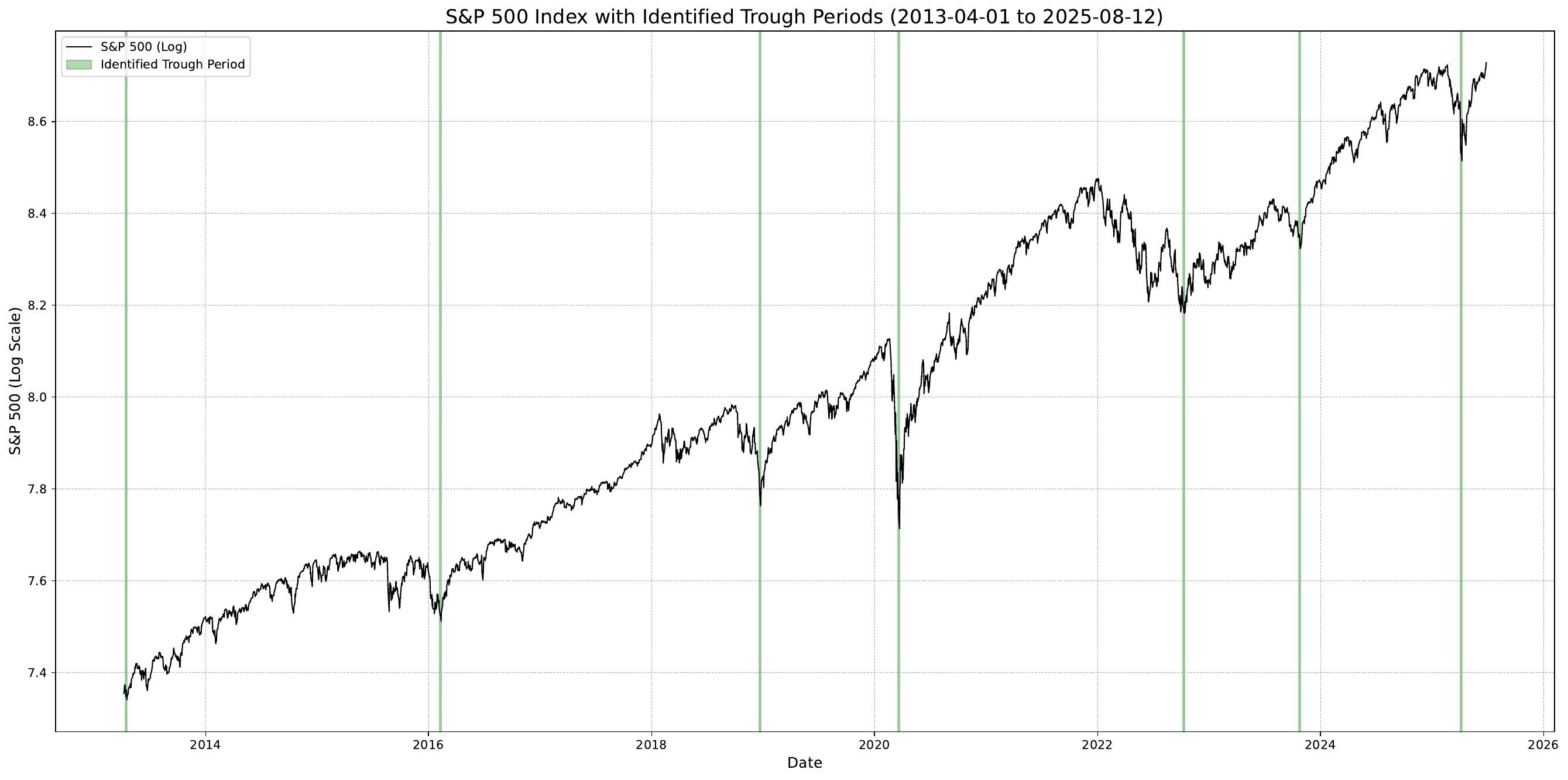}
    \caption{S\&P 500 Log Price and Identified Market Troughs (April 2013 - June 2025)}
    \label{fig:trough_viz}
    \justify
    \small{\textit{Notes:} This figure plots the daily log price of the S\&P 500 index over our sample period. The vertical green lines indicate the dates of significant market troughs. These troughs are identified using a pure implementation of the modified Bry-Boschan algorithm, as detailed in Section 2.2. To ensure all turning points within the sample period are identified algorithmically without end-of-sample ambiguity, the algorithm is applied to a data series extending beyond June 2025. The S\&P 500 price data is from Databento and Shiller's public database.}
\end{figure}

\subsection{Indicator Formation}

Using raw financial data, we construct a large and diverse set of over 200 possible predictors. To assist our analysis and economic interpretation, we group the indicators into two broad categories. The first group, which we call physical or structural indicators ($\mathbf{z}_t$), are intended to identify the underlying mechanics of the market (e.g., dealer positioning, monetary condition, liquidity). They represent actual flow and constraints in the financial system. Table \ref{tab:physical_features} details the definitions, mathematical expressions and economic rationale for the most relevant structural indicators. In addition to the indicators established from the literature, we also formulate a number of new metrics to capture economic characteristics not fully explained by standard measures. For instance, to measure the persistence and uni-directional nature of recent order flow, which is a potential sign of capitulation, we define a Flow Concentration measure. Similarly, we construct a measure for Unrealized Profits to measure the financial stress of recent market participants. The complete definitions can be found in Table \ref{tab:physical_features}. The second group of indicators, psychological or sentiment indicators ($\mathbf{u}_t$), is designed to quantify market fear, risk-seeking, and panic, that frequently reach extremes before market troughs. We present the complete definitions, mathematical expressions, and economic rationale for each of these indicators in Table \ref{tab:psych_features}.

To enhance the robustness of our study, we conduct a systematic treatment of outliers. We remove any value of Gamma Exposure indicators ($GEX_{OI}$ and $GEX_{V}$) exceeding the 99.9th percentile threshold. In the case of the open interest-based Put/Call Ratio ($PCR_{OI}$), we treat any observation of exactly zero as data artifacts and remove it from the analysis.

\begin{table}[tbp]
\centering
\caption{Physical/Structural Indicators ($\mathbf{z}_t$) and Economic Rationale}
\label{tab:physical_features}
\resizebox{\textwidth}{!}{%
\begin{tabular}{@{}lp{5.5cm}p{9cm}l@{}}
\toprule
\textbf{Name} & \textbf{Mathematical Definition} & \textbf{Economic Intuition} & \textbf{Reference(s)} \\
\midrule
GEX (OI) & $\sum_{i} (\Gamma_{C,i} \cdot \mathbf{OI}_{C,i} - \Gamma_{P,i} \cdot \mathbf{OI}_{P,i}) \times 100$ & Measures dealer gamma exposure from open positions. High positive GEX may suppress volatility, while low or negative GEX can amplify it. Capitulation troughs often occur in negative gamma regimes. & SqueezeMetrics \\
GEX (Volume) & $\sum_{i} (\Gamma_{C,i} \cdot \mathbf{V}_{C,i} - \Gamma_{P,i} \cdot \mathbf{V}_{P,i}) \times 100$ & Measures dealer gamma exposure from the day's trading volume, capturing intraday hedging pressures. &  \\
Delta Exposure & $\sum_{i} (\Delta_{C,i} \cdot \mathbf{OI}_{C,i} + \Delta_{P,i} \cdot \mathbf{OI}_{P,i}) \times 100$ & Measures net market delta positioning. Extremely low or negative values indicate bearish positioning and potential for short covering, often seen near troughs. & SqueezeMetrics \\
Order Flow Imbalance & \begin{tabular}[c]{@{}l@{}}$\sum_{k=1}^{N} \text{sign}(\mathbf{C}_k - \mathbf{O}_k) \cdot \text{Vol}_k$\\ \small{over 1-min bars}\end{tabular} & Proxy for net buying/selling pressure. Sustained, large negative OFI indicates aggressive selling that may precede seller exhaustion at a trough. & \citet{Easley2012} \\
Flow Concentration & $(\sum_{i=0}^{9} \mathbf{OFI}_{t-i}) \cdot \frac{\|\sum_{i=0}^{9} \mathbf{OFI}_{t-i}\|}{\sum_{i=0}^{9} \|\mathbf{OFI}_{t-i}\|}$ & Measures the persistence and unidirectionality of order flow. High negative values suggest sustained, concentrated selling, a hallmark of capitulation. &  \\
Unrealized Profit & $\frac{\mathbf{P}_t - \text{VWAP}_{63d}}{\text{VWAP}_{63d}}$ & Gauges the average unrealized profit/loss of market participants over a quarter. Large negative values mean recent participants are heavily underwater, increasing the odds of forced selling and a climax low. &  \\
Credit Spread & $\mathbf{Yld}_{HY} - \mathbf{Yld}_{RF}$ & The premium for bearing credit risk. A widening spread signals deteriorating economic conditions and heightened risk aversion, which peaks near market troughs. & \citet{Fama1989} \\
Amihud Illiquidity & $\frac{\|\mathbf{R}_{daily}\|}{\mathbf{V}_{\text{\$, daily}}}$ & Measures price impact. High values indicate illiquidity, as small volumes cause large price changes. Liquidity often vanishes near troughs. & \citet{Amihud2002} \\
FFR Slope & $\mathbf{P}_{C1} - \mathbf{P}_{C3}$ & Spread between 1st and 3rd Fed Funds futures. A steepening (more positive slope) can signal expectations of easier future policy, often a response to market stress. &  \\
FFR Basis & $(100 - \mathbf{P}_{C1}) - \text{EFFR}$ & The spread between the front-month implied Fed Funds rate and the spot effective rate. A positive basis indicates expected rate hikes or funding stress. &  \\
\bottomrule
\end{tabular}%
}
\justify
\small{\textit{Notes:} This table details a selection of the key physical and structural indicators used as predictive features in the analysis. These indicators are engineered to capture market microstructure, dealer positioning, and macroeconomic conditions that are less directly tied to immediate sentiment. All indicators are computed on a daily frequency for the full sample period from April 2013 to June 2025. The final features used in the model are transformations of these parent indicators, as described in Section 2.5.}
\end{table}

\begin{table}[tbp]
\centering
\caption{Psychological/Sentiment Indicators ($\mathbf{u}_t$) and Economic Rationale}
\label{tab:psych_features}
\resizebox{\textwidth}{!}{%
\begin{tabular}{@{}lp{5.5cm}p{9cm}l@{}}
\toprule
\textbf{Name} & \textbf{Mathematical Definition} & \textbf{Economic Intuition} & \textbf{Reference(s)} \\
\midrule
Realized Volatility & $\sqrt{252 \cdot (\sum_{i=1}^{M-1} r_{i,intra}^2 + r_{overnight}^2)}$ & Historical volatility from high-frequency data. Spikes in RV indicate panic and forced liquidation, which characterize market bottoms. & \citet{Andersen2003} \\
VIX & CBOE VIX Index methodology & Market's expectation of 30-day implied volatility. High VIX signals fear and demand for portfolio insurance, peaking at market troughs. & \citet{Whaley2000} \\
Volatility Risk Premium & $\mathbf{VIX}_t - \mathbf{RV}_t$ & The premium investors pay for protection against volatility. A negative VRP (realized > implied) often signals panic and deleveraging, a common feature of troughs. & \citet{Bollerslev2009} \\
PCR (OI) & $\frac{\sum \text{Put } \mathbf{OI}}{\sum \text{Call } \mathbf{OI}}$ & Ratio of open put to call contracts. High values indicate extreme bearish sentiment and hedging, which often precedes a market reversal. & \citet{Billingsley1988} \\
PCR (Volume) & $\frac{\sum \text{Put } \mathbf{Volume}}{\sum \text{Call } \mathbf{Volume}}$ & Ratio of traded put to call volume. Spikes indicate intense intraday fear and panic buying of puts, characteristic of capitulation lows. & \citet{Pan2006} \\
RN Skewness & $\mathbb{E}_{Q}[(\frac{K - \mu_K}{\sigma_K})^3]$ & Third moment of the risk-neutral distribution. Highly negative skew indicates high demand for OTM puts (crash protection), which is most pronounced at bottoms. & \citet{Bakshi2003} \\
RN Kurtosis & $\mathbb{E}_{Q}[(\frac{K - \mu_K}{\sigma_K})^4]$ & Fourth moment of the risk-neutral distribution. High kurtosis ("fat tails") indicates the market is pricing in a high probability of extreme moves. & \citet{Bakshi2003} \\
FX Momentum (EUR) & $\frac{\mathbf{P}_t - \mathbf{P}_{t-21}}{\mathbf{P}_{t-21}}$ & 21-day change in EUR/USD futures. Strong negative momentum (dollar strength) can reflect a "flight to safety" that accompanies equity market stress. & \citet{Asness2013} \\
FX Momentum (JPY) & $\frac{\mathbf{P}_t - \mathbf{P}_{t-21}}{\mathbf{P}_{t-21}}$ & 21-day change in JPY/USD futures. Strong positive momentum (yen strength) often reflects risk-off sentiment and carry trade unwinds during market turmoil. & \citet{Asness2013} \\
FX RV (EUR) & $\text{StDev}(\text{log}(\mathbf{P}_t/\mathbf{P}_{t-1}))_{21d} \cdot \sqrt{252}$ & 21-day rolling realized volatility of EUR/USD futures. Elevated volatility in major currency pairs often coincides with broad market deleveraging. & \citet{Andersen2003} \\
FX RV (JPY) & $\text{StDev}(\text{log}(\mathbf{P}_t/\mathbf{P}_{t-1}))_{21d} \cdot \sqrt{252}$ & 21-day rolling realized volatility of JPY/USD futures. Spikes in yen volatility are strongly associated with global risk-off events. & \citet{Andersen2003} \\
\bottomrule
\end{tabular}%
}
\justify
\small{\textit{Notes:} This table details a selection of the key psychological and sentiment indicators engineered for the predictive model. These indicators are designed to capture market fear, risk appetite, and panic, which often reach extreme levels near market troughs. All indicators are computed on a daily frequency for the full sample period from April 2013 to June 2025. The final features used in the model are transformations of these parent indicators, as described in Section 2.5.}
\end{table}

\subsection{Descriptive Statistics}

Table \ref{tab:desc_stats_parent} highlights summary statistics for the primary engineered parent indicators, and they provide insights that are important in our model design. Many series reveal non-normality; for example, the kurtosis of Gamma Exposure ({$gex_{oi}$}) is 1790.7, and the kurtosis for Realized Volatility ($RV$) is 32.4. Additionally, many series, like the credit spread and the VIX, are persistent: the first-order autocorrelation coefficients ($\rho(1)$) are very high, at 0.998 and 0.970 respectively. The fat tails and strong persistence properties of raw financial series motivate us to apply non-parametric scaling and non-linear machine models later.

\begin{table}[tbp]
\centering
\caption{Descriptive Statistics for Parent Indicators (2013-2025)}
\label{tab:desc_stats_parent}
\resizebox{\textwidth}{!}{
\begin{tabular}{lrrrrrrr}
\toprule
\textbf{Indicator} & \textbf{Mean} & \textbf{Std. Dev.} & \textbf{Skewness} & \textbf{Kurtosis} & \textbf{Min} & \textbf{Max} & \textbf{$\rho(1)$} \\
\midrule
\multicolumn{8}{l}{\textit{Panel A: Physical/Structural}} \\
\text{gex\_oi} & 6.65e+04 & 2.21e+06 & 41.417 & 1790.665 & -8.14e+06 & 1.03e+08 & 0.682 \\
\text{gex\_volume} & 6.43e+04 & 1.81e+06 & 30.648 & 1001.631 & -6.54e+06 & 6.47e+07 & 0.013 \\
\text{dex\_oi} & 3.06e+07 & 7.07e+07 & -0.795 & 5.632 & -4.45e+08 & 4.36e+08 & 0.943 \\
\text{ofi} & -4373.653 & 9.46e+04 & -0.582 & 3.428 & -6.70e+05 & 3.79e+05 & 0.054 \\
\text{credit\_spread} & 0.049 & 0.016 & 0.258 & 0.024 & 0.014 & 0.114 & 0.998 \\
\text{amihud\_illiquidity} & 9.94e-12 & 3.47e-11 & 0.000 & 0.000 & 0.000 & 1.17e-09 & -0.049 \\
\text{ffr\_slope} & 0.066 & 0.237 & 1.807 & 6.425 & -0.615 & 1.203 & 0.995 \\
\text{ffr\_basis} & 0.003 & 0.044 & 11.005 & 152.079 & -0.128 & 0.715 & 0.900 \\
\midrule
\multicolumn{8}{l}{\textit{Panel B: Psychological/Sentiment}} \\
\text{RV} & 12.700 & 9.840 & 4.105 & 32.356 & 0.758 & 133.842 & 0.669 \\
\text{VIX} & 17.812 & 6.942 & 2.730 & 13.973 & 9.140 & 82.690 & 0.970 \\
\text{VRP} & 3.336 & 5.042 & -3.574 & 30.486 & -58.725 & 16.729 & 0.662 \\
\text{PCR\_OI} & 1.819 & 0.185 & 0.178 & -0.554 & 1.389 & 2.489 & 0.987 \\
\text{PCR\_V} & 1.390 & 0.318 & 0.356 & 0.335 & 0.533 & 3.092 & 0.729 \\
\bottomrule
\end{tabular}
}
\justify
\small{\textit{Notes:} This table reports summary statistics for the untransformed "parent" indicators at a daily frequency for the sample period April 2013 to June 2025. The final column, $\rho(1)$, reports the first-order autocorrelation coefficient. The pronounced non-normality (e.g., kurtosis of 1790 for \text{gex\_oi}) and high persistence (e.g., $\rho(1) > 0.9$ for VIX and credit spreads) motivate the feature transformations and use of nonlinear models detailed in Sections 2.5 and 3.}
\end{table}

\subsection{Advanced Feature Engineering and Scaling}

The features generated are subject to further transformation. For any input time series $\mathbf{X}_t$, we calculate the Rate-of-Change (ROC)\footnote{The acronym ROC is used throughout this paper in our feature names (e.g., `\_roc63\_`) to indicate Rate-of-Change. It should be noted that this acronym should not be confused with the Receiver Operating Characteristic (ROC) curve for evaluating models, which we refer to as the ROC AUC.}, Trend Z-Score, and Wavelet Decomposition components. Lastly, we apply a rolling percentile rank transformation across a 252-day window to all features, which converts the value of each feature into the interval $[-1,1]$. Thus, we have a robust, non-parametric representation of each feature's value with respect to its recent history.

\section{Predictive Modeling Framework}

\subsection{Problem Formulation and Labeling}

We formulate the task as a binary classification objective, and we illustrate the mechanism in Figure \ref{fig:labeling_methodology}. For any trough date $T$, a positive label ($\mathbf{y}_t=1$) is assigned to all time steps $t$ in the $W_L$-day labeling window immediately before the trough, such that $t \in [T-W_L, T]$. All other days are assigned a negative label ($\mathbf{y}_t=0$). Thus, the objective for the model is to predict $P(\mathbf{y}_t=1 | \mathcal{X}_t)$, i.e. the probability that day $t$ is an ex-post confirmed trough timeperiod, using only the historical feature data that exists in $\mathcal{X}_t$. In order to incorporate the temporal dynamics, the input for a prediction is a tensor $\mathcal{X}_t \in \mathbb{R}^{L \times D}$, corresponding to the $D$ feature vectors from the last $L$ time steps, which assembles to the lookback window.

\begin{figure}[htbp]
    \centering
    \includegraphics[width=\columnwidth]{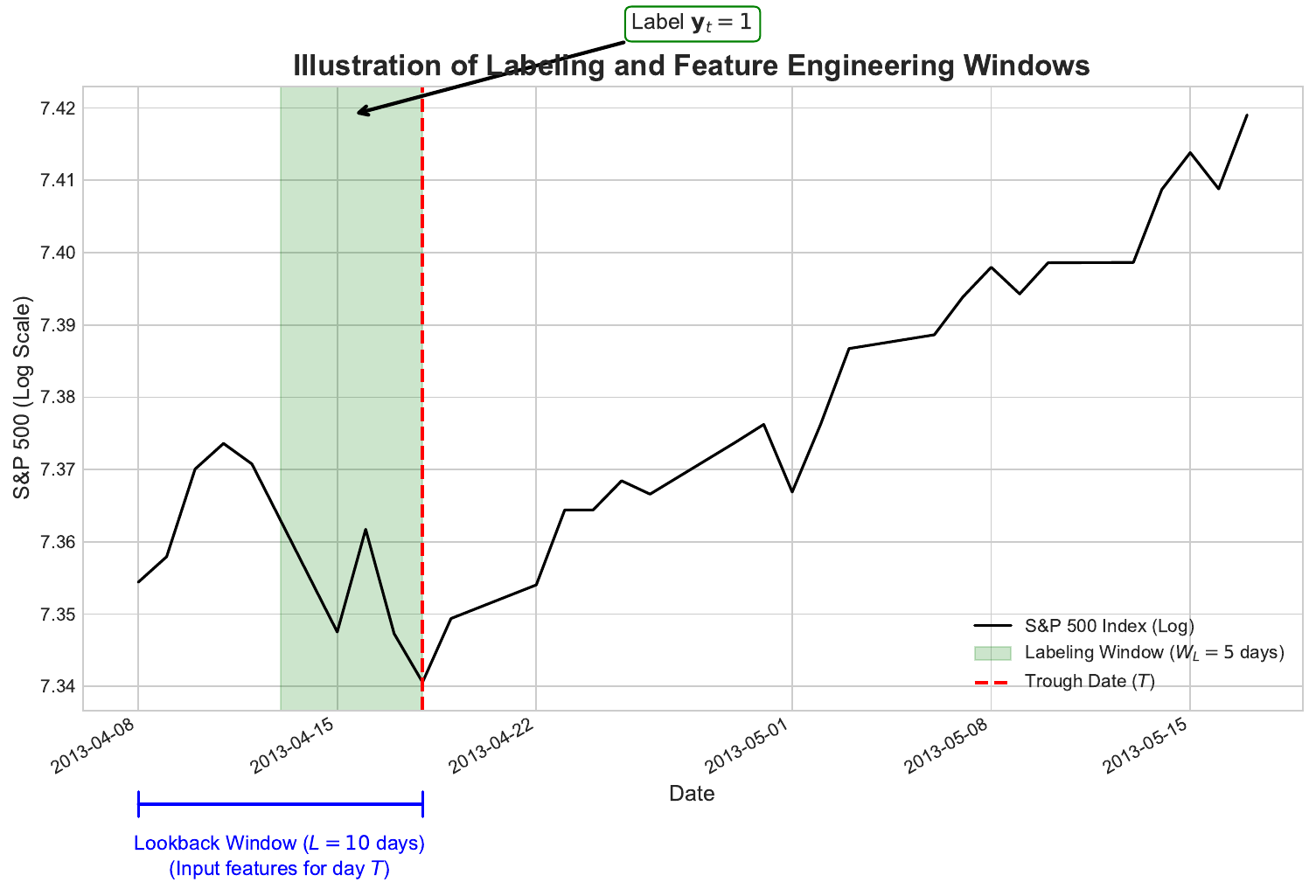}
    \caption{Illustration of the Labeling and Feature Engineering Methodology. 
    The figure plots the daily log price of the S\&P 500 Index for an illustrative period around the April 18, 2013 market trough. The trough date is denoted as $T$. For our classification task, a positive label ($\mathbf{y}_t=1$) is assigned to all days within the shaded green \textit{labeling window} ($W_L=5$ days), defined as the period $t \in [T-W_L, T]$. The model's prediction for any given day uses input features derived from the data in the preceding blue \textit{lookback window} ($L=10$ days).}
    \label{fig:labeling_methodology}
\end{figure}

\subsection{Feature Aggregation and Stationarity}

The sequence tensor $\mathcal{X}_t$ is aggregated to a feature vector $\mathbf{x}_t \in \mathbb{R}^{4D}$ by calculating four statistics for each of the $D$ features over the lookback window $L$: Mean, Standard Deviation, Trend (slope of linear regression), and Last Value. This is important as it converts non-stationary indicators into stationary features. An Augmented Dickey-Fuller (ADF) test demonstrated that the percentage of stationary features increased from 90.6\% in the parent set to 100\% in the final aggregated set, increasing the stability of the model.

\subsection{Model Training and Hyperparameter Tuning}

We used a nested cross-validation pipeline on a `TimeSeriesSplit` of the data to select the best model while avoiding data leakage. The inner loop is for hyperparameter tuning, while the outer loop provides an unbiased estimate of generalization performance. The pipeline within each fold is:

\begin{enumerate}

    \item \textbf{Data Augmentation}: We use SMOTE (Synthetic Minority Over-sampling Technique) to the training set, given the severe class imbalance.

    \item \textbf{Feature Scaling}: We fit a `StandardScaler` only on augmented training data.

    \item \textbf{Feature Selection}: We train a Random Forest classifier, and selected the top $N$ features from the training data based on Gini importance.

    \item \textbf{Model Fitting}: We train an SVM to the final processed training data.

\end{enumerate}

Table \ref{tab:hyperparams} presents the search space and final values determined through these steps. We selected hyperparameters using the one-standard-error rule\footnote{The one-standard-error rule is a heuristic for choosing a parsimonius model, which has statistically similar predictive performance to the most optimal model determined through cross-validation. The procedures for the one-standard-error rule are as follows. The candidate model with the maximum mean score is identified and its standard error is computed. From all candidate models whose mean score is within one standard error of the best mean score, we choose the simplest model. In our case, the prefered SVM model is specifically the one with the least features (linear kernel, and lowest C, or regularization parameter.)} to balance predictive performance against model complexity. For the final production model (trained on the entire main dataset for evaluation using the hold-out test set), the pipeline resulted in identifying the 15 features listed in Table \ref{tab:final_features}.

\begin{table}[htbp]
\centering
\caption{Hyperparameter Tuning and Final Values}
\label{tab:hyperparams}
\resizebox{\columnwidth}{!}{%
\begin{tabular}{lcc}
\toprule
\textbf{Hyperparameter} & \textbf{Search Space} & \textbf{Final Value} \\
\midrule
Labeling Window $W_L$ (days) & \{5, 10, 15, 20, 25, 30\} & 5 \\
Lookback Window $L$ (days) & \{5, 10, 15, 20, 25, 30\} & 10 \\
Number of Features $N_{features}$ & \{10, 15, 20, 25, 30\} & 15 \\
Augmentation Method & \{none, SMOTE, jitter, mixup\} & SMOTE \\
Oversample Factor & \{0.5, 1.0, 1.5\} & 1.0 \\
SVM Kernel & \{linear, rbf\} & linear \\
SVM C (Regularization) & \{0.01, 0.1, 1.0\} & 0.01 \\
\bottomrule
\end{tabular}%
}
\flushleft
\small{\textit{Notes:} This table presents the hyperparameter search space and the optimal values selected for the primary SVM classification model. The selection is performed using a nested time-series cross-validation procedure on the training data. The final values are chosen to maximize the out-of-fold ROC AUC score, applying the one-standard-error rule to select the simplest model within one standard error of the best-performing model.}
\end{table}

\begin{table}[htbp]
\centering
\caption{Final 15 Features for the Production Model}
\label{tab:final_features}
\begin{tabular}{l}
\toprule
\textbf{Feature Name} \\
\midrule
\seqsplit{vix\_wave\_cA3\_scaled\_last} \\
\seqsplit{gex\_oi\_wave\_cA3\_scaled\_mean} \\
\seqsplit{upg\_63d\_scaled\_last} \\
\seqsplit{credit\_spread\_roc63\_scaled\_std} \\
\seqsplit{dex\_oi\_wave\_cA3\_scaled\_mean} \\
\seqsplit{dex\_oi\_wave\_cA3\_scaled\_last} \\
\seqsplit{realized\_volatility\_wave\_cA3\_scaled\_last} \\
\seqsplit{upg\_63d\_wave\_cA3\_scaled\_last} \\
\seqsplit{vix\_scaled\_mean} \\
\seqsplit{credit\_spread\_scaled\_last} \\
\seqsplit{gex\_oi\_roc63\_scaled\_std} \\
\seqsplit{realized\_volatility\_wave\_cA3\_scaled\_mean} \\
\seqsplit{pcr\_volume\_scaled\_std} \\
\seqsplit{upg\_63d\_wave\_cA3\_scaled\_mean} \\
\seqsplit{vix\_scaled\_last} \\
\bottomrule
\end{tabular}
\flushleft
\small{\textit{Notes:} This table lists the 15 features selected by a Random Forest classifier based on Gini importance. This selection is performed when the final production model is trained on the entire main dataset (Apr 2013 - Jun 2023), using the optimal hyperparameters found during cross-validation. The feature selection step is performed dynamically within each CV fold for model evaluation, meaning the feature sets used during CV vary. The list shown here represents the specific inputs for the final model that is evaluated on the hold-out test set.}
\end{table}

\subsection{Out-of-Sample Probability Calibration}

SVM raw scores are not well calibrated probabilities. To tackle this, we calibrate the scores into reliable probabilities using an IsotonicRegression model, which is trained on the out-of-sample predictions and true labels from each of the cross-validation folds, so as to prevent any data leakage during the calibration step. The upper panel in Figure \ref{fig:main_results} demonstrates the good calibration of the resulting nowcast on the hold-out test set.

\section{Empirical Results and Interpretations}

\begin{figure}[!t]
    \centering
    \includegraphics[width=0.85\textwidth]{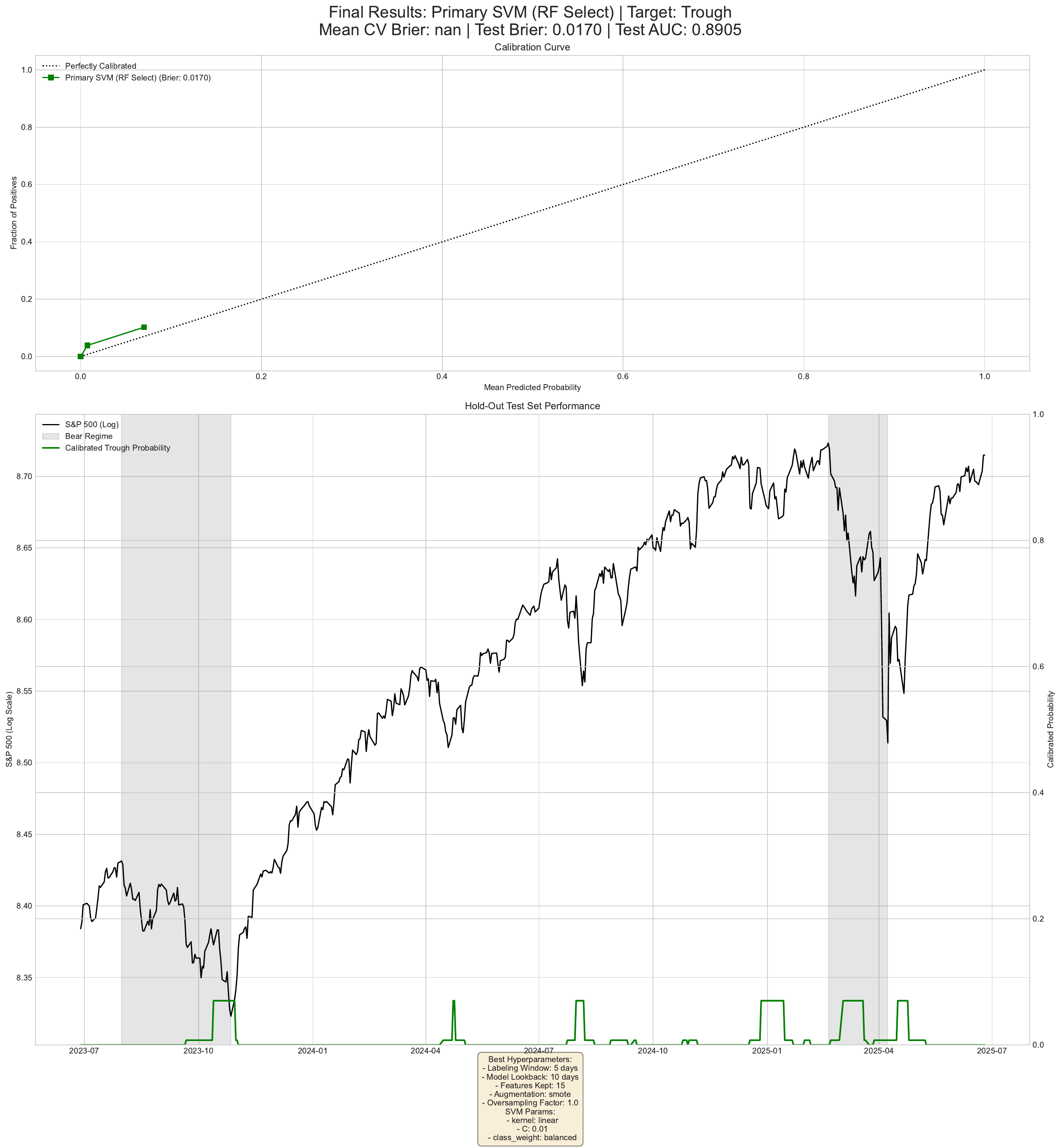}
    \caption{Out-of-Sample Predictive Performance and Trough Probabilities}
    \label{fig:main_results}
    \flushleft
    \small % Using \small for notes is common practice
    \textit{Notes:} This figure evaluates the primary model's performance on the hold-out test set from July 2023 to June 2025. The model is a Support Vector Machine (SVM) using 15 features selected via Random Forest Gini importance, with outputs calibrated post-hoc via Isotonic Regression (see Section 3.3 for details). 
    
    \textbf{Panel A (Top): Probability Calibration.} The panel plots the model's calibration curve. The dashed diagonal line represents perfect calibration. The solid line shows the model's reliability, achieving a Brier score of 0.0170. The model's overall discriminatory power, measured by the Area Under the ROC Curve (AUC), is 0.8905.
    
    \textbf{Panel B (Bottom): Predicted Trough Probability.} This panel plots the model's calibrated daily probability of being in a trough state (green line, right y-axis) against the S\&P 500 log price series (black line, left y-axis). The shaded vertical bars denote the actual market trough periods identified in Table \ref{tab:turning_points}. The hyperparameters for this model specification are listed in the embedded text box.
\end{figure}

We evaluate our primary model on the hold-out test set, present its performance against a number of benchmarks, and provide a detailed interpretation of the model predictions and drivers.

\subsection{Model Performance Evaluation}

\subsubsection{Decision on Performance Metrics}

Due to the extreme class imbalance in predicting rare market troughs, standard classification metrics are ill-posed for this problem. Precision, Recall, and F1-Score are all contingent on assigning a fixed decision threshold (e.g. 0.5) to a model's probabilistic output. When the positive class (a trough) is incredibly rare, we would expect a very well-calibrated model to assign a very low probability to this event on most days. That means the predicted probability is almost never going to cross the 0.5 mark, resulting in zero positive predictions. If the number of true positives is zero, Precision, Recall, and F1 all go to zero and the default threshold is not measuring the predictive power of the model itself, but the inappropriateness of using a default threshold. Thus, our assessment is based on two metrics that are threshold insensitive and directly assess the quality of the nowcasting market capitulation warning system:

\begin{itemize}

    \item \textbf{ROC AUC}: The area under the receiver operator characteristic curve quantifies how well a model \textit{ranks} the observations correctly. Specifically, if we are to randomly present both a positive instance and a negative instance, how often would the model rank the positive higher? A high AUC value indicates good discrimination.

    \item \textbf{Brier Score}: The Brier score measures the \textit{accuracy and calibration} of the probability forecast itself. It is the mean-squared error of each of the predicted probabilities against the truth of what happened (0 or 1). A lower Brier score means that the model's probability output is more trustworthy and closer to the true likelihood of the event.

\end{itemize}

We choose ROC AUC to measure discrimination, and Brier score to measure probabilistic reliability, because they provided the best representation of the practical value of the model.

\subsubsection{Model Performance and Benchmarks}

Our primary model shows strong capability with out-of-sample predictability and reliability on the hold-out test set. It attained a ROC AUC of \textbf{0.8905}, showing strong discriminatory power, and a Brier score of \textbf{0.0170}, indicating reliable probabilities. A visualized summary of the results is shown in Figure \ref{fig:main_results}. The top portion has the calibration curve, and the bottom shows practical utility of the model, showing that the green spikes in predicted trough probability act as the timely and accurate signal to actual market troughs. In non-capitulating stable market regimes, the model's green probability line remain around zero, demonstrating its ability to largely avoid false alarms.  

\FloatBarrier

We compare our nowcast model with a range of benchmarks in Table \ref{tab:model_performance}, which provides important context for our model's performance. The LassoCV model had the highest ROC AUC (0.9495), but the extremely poor Brier score (0.2528) underscores that the model's raw outputs are completely uncalibrated/invalid as probabilities; this highlights the merit of the post-hoc calibration done in our primary pipeline. The naive heuristic (VIX > 40) had low discriminatory power (AUC of 0.6656), confirming that our framework provides better utility. Lastly, the Gaussian Naive Bayes model performed worse than random guessing (AUC < 0.5), indicating that the core idea of conditional independence that governs the Naive Bayes model is violated. Thusly, our primary SVM model provided the best convergence of high discriminatory power combined with trustworthy probability nowcasting.

\begin{table}[htbp]
\centering
\caption{Out-of-Sample Performance Comparison on the Hold-Out Test Set}
\label{tab:model_performance}
\begin{tabular}{lcc}
\toprule
\textbf{Model} & \textbf{ROC AUC} & \textbf{Brier Score} \\
\midrule
\textbf{Primary SVM (RF Select)} & \textbf{0.8905} & \textbf{0.0170} \\
\midrule
\textit{Benchmark Models} \\
Vanilla SVM (All Features) & 0.9061 & 0.0176 \\
LassoCV & 0.9495 & 0.2528 \\
Heuristic (VIX > 40) & 0.6656 & 0.0140 \\
Gaussian Naive Bayes & 0.4878 & 0.0180 \\
\bottomrule
\end{tabular}
\flushleft
\small{\textit{Notes:} This table compares the out-of-sample performance of the primary SVM model against several benchmarks on the hold-out test set (July 2023 - June 2025). Performance is measured by the Area Under the ROC Curve (ROC AUC), which assesses discriminatory power (higher is better), and the Brier Score, which measures the accuracy of probability forecasts (lower is better). The "Primary SVM (RF Select)" model is the main model from Section 3, with features selected by a Random Forest. The benchmark models are trained and evaluated under an identical time-series cross-validation framework for a fair comparison. The "Heuristic (VIX > 40)" is a simple rule-based benchmark. The poor Brier score of the LassoCV model highlights its lack of probability calibration.}
\end{table}

\subsection{Interpretation of Model Predictions}

\subsubsection{Feature Importance with SHAP}
To understand which factors drive the model's predictions, we employ SHAP (SHapley Additive exPlanations) \citep{Lundberg2017}. Figure \ref{fig:shap_bar} provides a global summary by ranking features based on their mean absolute SHAP value, which represents their average impact on the model's output magnitude. The plot clearly identifies \seqsplit{gex\_oi\_roc63\_scaled\_std} (the standard deviation of the 63-day rate-of-change in Gamma Exposure) and \seqsplit{credit\_spread\_roc63\_scaled\_std} as the two most influential predictors.

To provide further insight into the characteristics of these key drivers, Table \ref{tab:desc_stats_agg} presents descriptive statistics for the five most important features identified by our SHAP analysis. The statistics show that our feature engineering and scaling process has successfully created well-behaved inputs for the model; compared to the raw parent indicators in Table \ref{tab:desc_stats_parent}, these final features have much lower skewness and kurtosis. Their high first-order autocorrelation, with $\rho(1)$ values exceeding 0.9 for most, confirms the persistent, trend-like nature of the signals the model has learned to rely on.

To understand the directionality and heterogeneity of these impacts, Figure \ref{fig:shap_beeswarm} visualizes the SHAP value for every individual prediction in our hold-out set. The interpretation reveals nuanced relationships. Examining the top feature, \seqsplit{gex\_oi\_roc63\_scaled\_std}, we observe a pattern that is consistent with the figure's caption: high values of this feature (red dots) are associated with negative SHAP values, meaning they push the prediction toward a lower probability of a trough. Conversely, low values of this feature (blue dots) have a neutral or positive impact. This suggests the model has learned that a trough is more probable not when GEX is changing chaotically, but when its rate-of-change is smoother and more persistent. For the second feature, \seqsplit{credit\_spread\_roc63\_scaled\_std}, high values (red dots) have positive SHAP values, confirming that rising volatility in credit spreads is a key indicator of market stress that contributes to the model's trough predictions.

\begin{table}[tbp]
\centering
\caption{Descriptive Statistics for Key Predictive Features}
\label{tab:desc_stats_agg}
\resizebox{\textwidth}{!}{
\begin{tabular}{lrrrrrrr}
\toprule
\textbf{Feature} & \textbf{Mean} & \textbf{Std. Dev.} & \textbf{Skewness} & \textbf{Kurtosis} & \textbf{Min} & \textbf{Max} & \textbf{$\rho(1)$} \\
\midrule
`gex\_oi\_roc63\_scaled\_std` & 0.306 & 0.146 & 0.457 & 0.020 & 0.012 & 0.914 & 0.927 \\
`credit\_spread\_roc63\_scaled\_std` & 0.130 & 0.099 & 1.662 & 4.337 & 0.002 & 0.707 & 0.954 \\
`realized\_volatility\_wave\_cA3\_scaled\_last` & -0.058 & 0.618 & 0.121 & -1.270 & -0.992 & 1.000 & 0.986 \\
`vix\_wave\_cA3\_scaled\_last` & -0.063 & 0.634 & 0.082 & -1.305 & -0.992 & 1.000 & 0.992 \\
`upg\_63d\_scaled\_last` & -0.027 & 0.597 & 0.050 & -1.212 & -0.992 & 1.000 & 0.950 \\
\bottomrule
\end{tabular}
}
\flushleft
\small{\textit{Notes:} This table presents summary statistics for the five most important final features used in the predictive model, as determined by the global SHAP analysis shown in Figure \ref{fig:shap_bar}. The statistics are calculated over the full sample period from April 2013 to June 2025, which comprises $N=3068$ daily observations. These "final aggregated features" are transformations (e.g., standard deviation, wavelet component) of the parent indicators from Table \ref{tab:desc_stats_parent} and have been scaled to the interval [-1, 1]. The final column, $\rho(1)$, is the first-order autocorrelation coefficient, indicating high persistence in these key predictive signals.}
\end{table}

\begin{figure}[htbp]
    \centering
    \includegraphics[width=\columnwidth]{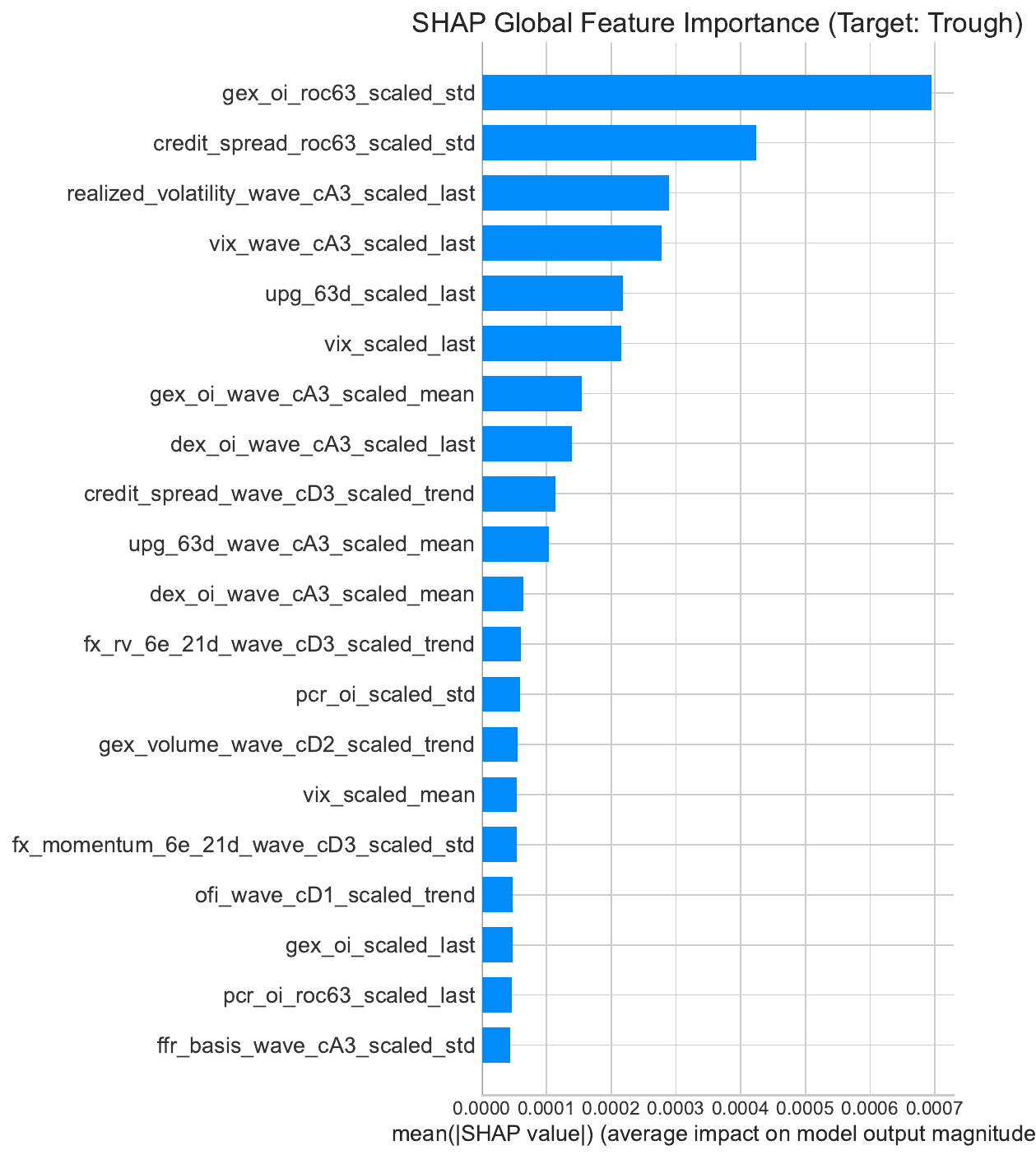}
    \caption[SHAP Global Feature Importance for the SVM Model]{SHAP Global Feature Importance for the SVM Model on the Hold-Out Test Set. \newline \textit{Notes:} The figure displays the mean absolute SHAP (SHapley Additive exPlanations) value for the top 20 features from our primary SVM classification model, evaluated on the hold-out test set (2022-2025). The x-axis represents the average magnitude of a feature's impact on the model's log-odds output for predicting a market trough. Features are ranked in descending order of importance. For instance, \seqsplit{gex\_oi\_roc63\_scaled\_std} has the largest average impact on the model's predictions.}
    \label{fig:shap_bar}
\end{figure}

\begin{figure}[htbp]
    \centering
    \includegraphics[width=0.8\columnwidth]{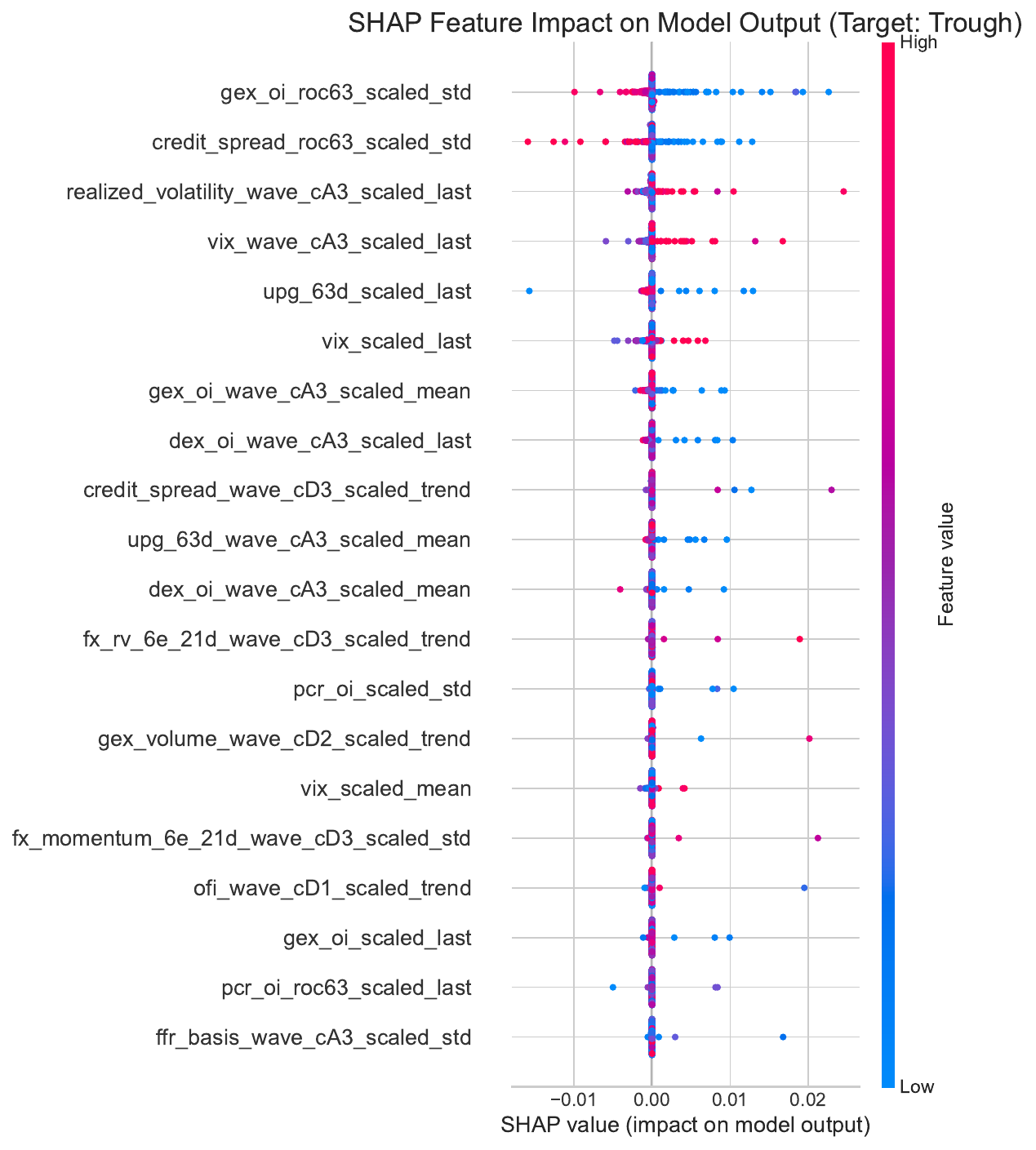}
    \caption[SHAP Feature Dependence Beeswarm Plot]{SHAP Feature Dependence Beeswarm Plot on the Hold-Out Test Set. \newline \textit{Notes:} The figure illustrates both the magnitude and direction of feature impacts on the SVM model's predictions for the hold-out test set. Each dot corresponds to a single daily observation for a given feature. The dot's horizontal position indicates its SHAP value—a positive value pushes the prediction towards a higher probability of a trough, while a negative value pushes it lower. The color represents the feature's normalized value for that day, from low (blue) to high (red). For the top feature, \seqsplit{gex\_oi\_roc63\_scaled\_std}, high values (red dots) are associated with negative SHAP values, indicating that high volatility in the rate-of-change of GEX makes a trough less likely. Conversely, low values (blue dots) are associated with neutral or positive SHAP values, suggesting a smooth, persistent change in GEX is more indicative of an approaching trough.}
    \label{fig:shap_beeswarm}
\end{figure}

\subsubsection{Feature Dependence and Interaction}
To move beyond global importance and explore nonlinear relationships, we examine SHAP dependence plots. These plots show how a feature's marginal contribution to the prediction (its SHAP value) changes across the range of its values. Figure \ref{fig:shap_dependence} illustrates these relationships for our most influential predictors, revealing key nonlinearities and interaction effects learned by the model.
\begin{itemize}
    \item Panel (a) shows the impact of the standard deviation of the rate-of-change in Gamma Exposure \seqsplit{gex\_oi\_roc63\_scaled\_std}. The model has learned a nuanced, nonlinear relationship. The feature's strongest positive impact on predicting a trough (i.e., the highest SHAP values) occurs when its value is in a low-to-moderate range. This effect is potently amplified by an interaction: the positive push towards a trough prediction happens almost exclusively when the level of GEX itself is low (indicated by the blue points for \seqsplit{gex\_oi\_scaled\_last}). This aligns with the economic intuition of a "negative gamma" regime, where dealer hedging amplifies downward moves. The model has learned that a trough is most probable not when GEX is changing chaotically (a high \seqsplit{\_std} value, which has a neutral impact), but rather when the market is in a low-gamma state and the change in GEX is exhibiting persistent, low-to-moderate volatility.
    \item Panel (b) reveals a powerful, non-monotonic interaction effect. The impact of credit spread volatility (\seqsplit{credit\_spread\_roc63\_scaled\_std}) on the prediction is entirely conditional on the state of the underlying market volatility trend (\seqsplit{realized\_volatility\_wave\_cA3\_scaled\_last}). The model has learned a "canary in the coal mine" signal where the feature's impact is positive but highly localized. The strongest push towards a trough (the highest positive SHAP values) occurs at very low levels of credit spread volatility (x-axis near 0.0-0.05), an effect that is present only when underlying market volatility is moderate (purple/magenta points). This positive impact then diminishes as credit spread volatility increases further. Conversely, if the market is already in a state of high underlying volatility (red points), then increased credit spread volatility consistently pushes the trough probability lower (negative SHAP values).
    \item In Panel (c), we observe a distinct threshold effect for \seqsplit{realized\_volatility\_wave\_cA3\_scaled\_last}. The feature has little impact on the prediction when its value is below approximately 0.6. However, beyond this point, its SHAP value increases sharply, indicating that very high levels of realized volatility are a strong signal of an impending trough. The coloring reveals a potent interaction: this effect is magnified when the trend in the Fed Funds basis (\seqsplit{ffr\_basis\_roc63\_scaled\_trend}) is low (blue points), suggesting that high market volatility is most dangerous when it coincides with deteriorating expectations for near-term funding conditions.
\end{itemize}

\begin{figure}[htbp]
    \centering
    \subfigure[GEX Open Interest Volatility]{
        \includegraphics[width=0.3\textwidth]{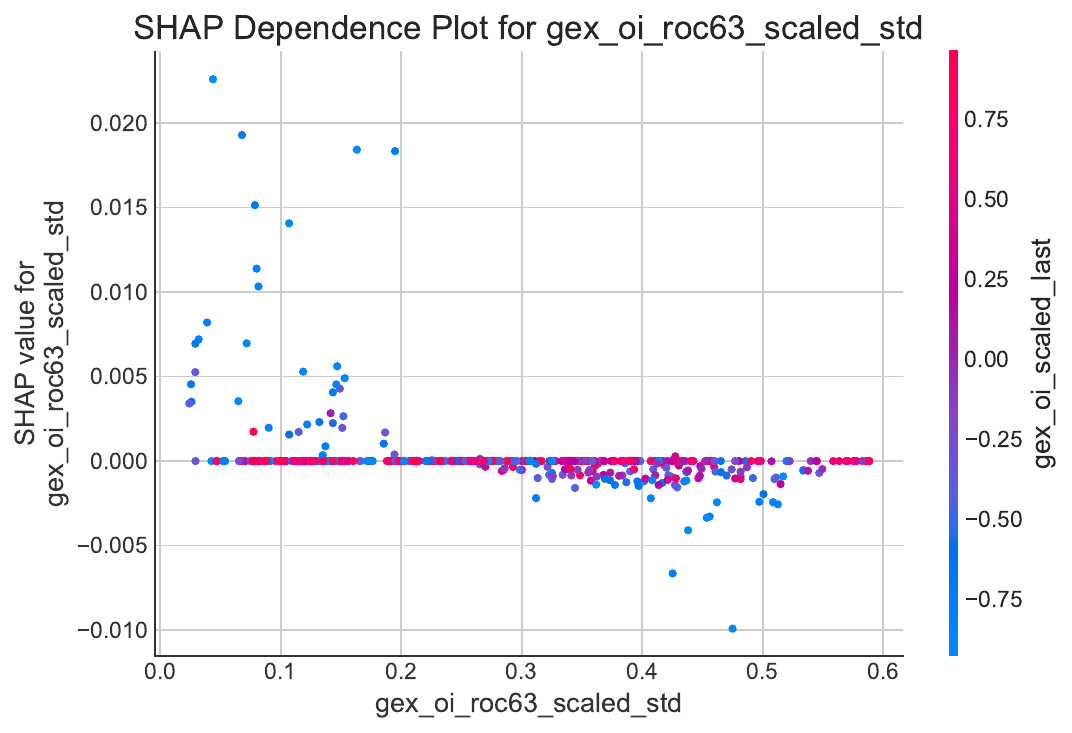}
        \label{fig:dep1}
    }
    \subfigure[Credit Spread Volatility]{
        \includegraphics[width=0.3\textwidth]{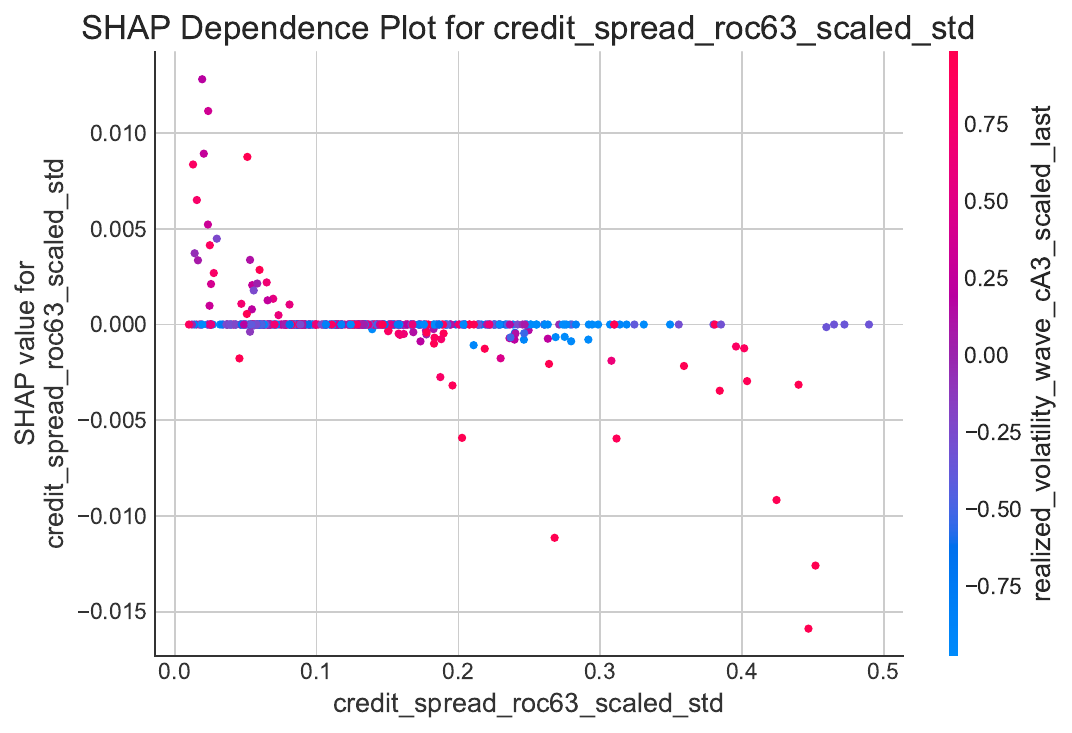}
        \label{fig:dep2}
    }
    \subfigure[Realized Volatility]{
        \includegraphics[width=0.3\textwidth]{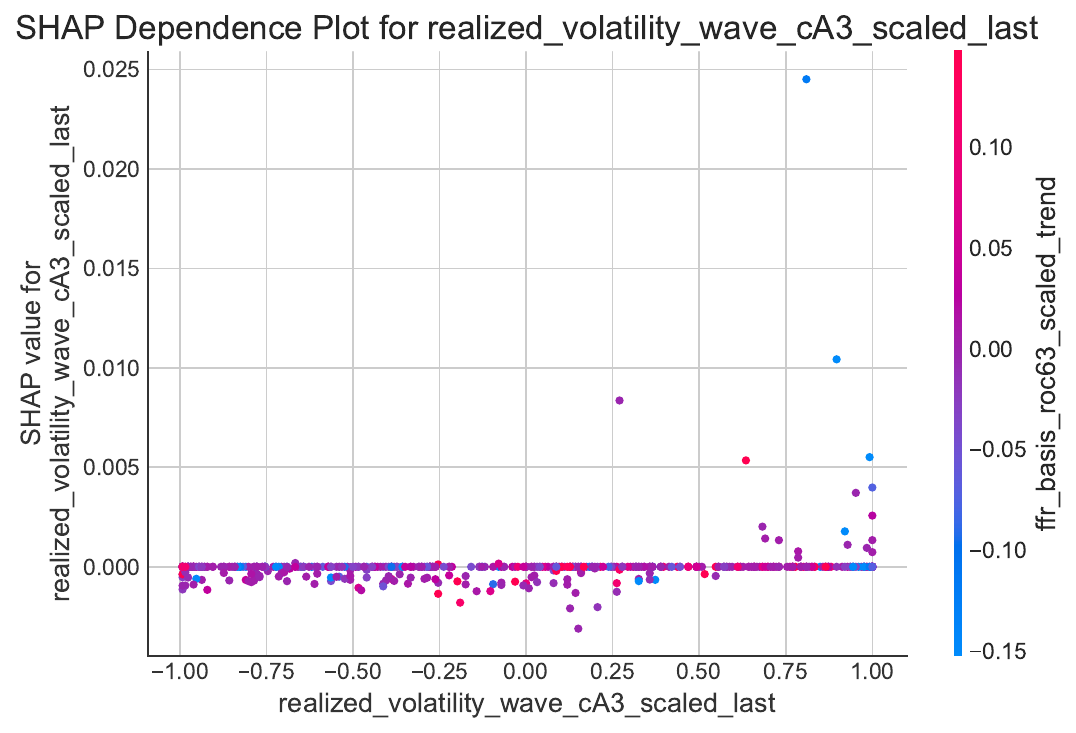}
        \label{fig:dep3}
    }
    \caption{SHAP Dependence Plots for Top Predictive Features}
    \label{fig:shap_dependence}
    \flushleft
    \textit{Notes:} The figure shows the relationship between a feature's value (x-axis) and its impact on the model's prediction in terms of its SHAP value (y-axis) for the primary SVM model. Each point represents a single observation from the hold-out test set. A positive SHAP value indicates the feature pushed the prediction towards a higher probability of a market trough. The points are colored by the value of a second feature, chosen automatically by the SHAP library to display the strongest interaction effects.
    (a) Plots the SHAP value for the standard deviation of the scaled GEX from open interest (`gex\_oi\_roc63\_scaled\_std`). The color corresponds to the last value of GEX (`gex\_oi\_scaled\_last`).
    (b) Plots the SHAP value for the standard deviation of the scaled credit spread (`credit\_spread\_roc63\_scaled\_std`). The color corresponds to the last value of wavelet-transformed realized volatility (`realized\_volatility\_wave\_cA3\_scaled\_last`).
    (c) 
    Plots the SHAP value for the last value of wavelet-transformed realized volatility (`realized\_volatility\_wave\_cA3\_scaled\_last`). The color corresponds to the trend in the Fed Funds basis (`ffr\_basis\_roc63\_scaled\_trend`).
\end{figure}

\section{Robustness and Stability Evaluation}

\label{sec:robustness}

One of the primary difficulties for any predictive model in finance is dealing with structural breaks: a fundamental change in the data-generating process in the market, which can disrupt relationships that are learned from historical data. Conventional econometric models (e.g., OLS, VAR) with fixed parameters are particularly susceptible to structural breaks, whereas our machine learning pipeline is more robust in being designed for flexibility. Our SVM model is non-parametric and implements adaptive feature engineering over rolling windows, further enhanced by the robust time-series cross-validation protocol, so it should be robust to evolving market environments.

Nonetheless, we conduct a series of diagnostic evaluations on the hold-out sample in order to justify the robustness of our model. These evaluations are used to identify common failure modes in machine learning models, such as degraded performance, covariate shift (when the distributions of input data change), and concept drift (when the relationship between inputs and outcome change).

\subsection{Stability of Model Performance Over Time}

\begin{itemize}

    \item \textbf{Rationale}: The simplest way to evaluate model robustness is to look at how performance changes over time. A stable model should be able to maintain its predictive power, while a model that has suffered from a structural break would demonstrate abrupt and sustained decline in performance. The Brier score is a useful metric for assessing performance, as it captures the accuracy of probabilistic forecasts.

    \item \textbf{Performance Evaluation}: We estimate the Brier score on the hold-out test set over a \textbf{63-day rolling window} (approximately one trading quarter), rather than estimating it as a single number. This allows us to assess the model's calibration and accuracy chronologically. \item \textbf{Results}: Our findings reveal that the model is highly stable. The rolling Brier score stays extremely low (near 0.00) for the vast majority of the test period, suggesting consistently measured probabilities that are accurate and well-calibrated during a stable market. The brier score shows two elevated periods that maps to the two actual market trough events identified by the BB algorithm. These spikes should not be interpreted as the failure of the model, but merely a reflection of the inherent difficulty and uncertainty of those specific moments. Importantly, the Brier score falls back to its low baseline quickly after the spikes, indicating that model performance is not predictably worse after a crisis event. This verifies the model is not "broken" by market capitulation; it recognizes it, and then stabilizes back to near zero, as illustrated in Figure \ref{fig:rolling_brier}.

\end{itemize}

\begin{figure}[htbp]
\centering
\includegraphics[width=\columnwidth]{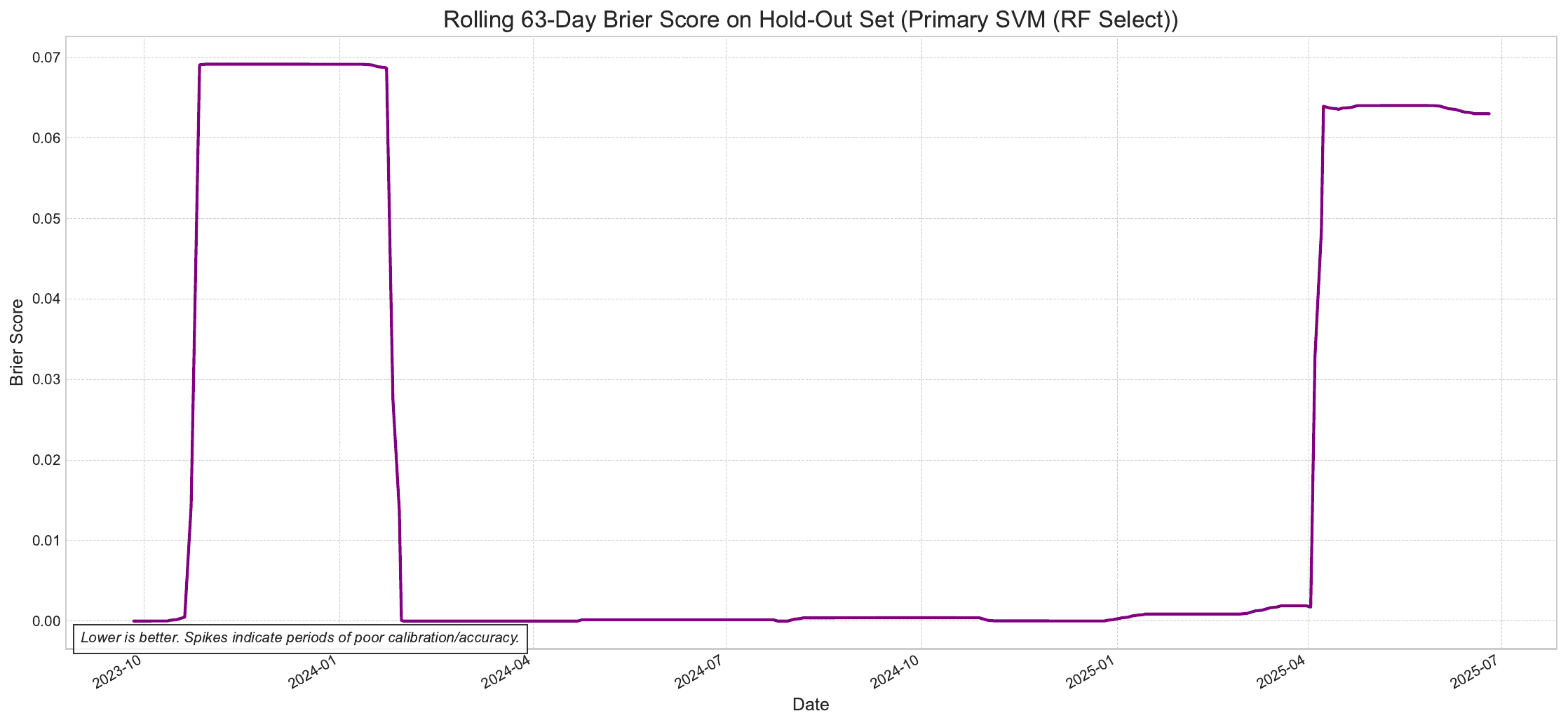}
\caption{Model Performance Stability on the Hold-Out Test Set}
\label{fig:rolling_brier}
\justify
\small{\textit{Notes:} This figure plots the Brier score of the primary SVM model's calibrated probability forecasts, calculated over a 63-day rolling window. The sample is the hold-out test set, covering the period from July 2023 to June 2025. The Brier score measures the mean squared error between predicted probabilities and actual outcomes; a lower score indicates better forecast accuracy and calibration. The sharp spikes in the score around October 2023 and April 2025 coincide with the actual market troughs identified in Table \ref{tab:turning_points}. The score's rapid return to a near-zero baseline following these events demonstrates that the model's performance is stable and does not persistently degrade after periods of market stress.}
\end{figure}

\subsection{Input Feature Stability: Covariate Shift Analysis}

\begin{itemize}

    \item \textbf{Rationale}: A model trained on data with a particular distribution perform poorly when it is asked to make predictions using data with a markedly different distribution---this is called covariate shift. To test for this, we compare the distributions of the most important input features selected by the SVM nowcasting model between the training and testing time periods.

    \item \textbf{Implementation}: we plot the kernel-density estimates (KDEs) for the five features with the highest Gini importance from the Random Forest selector (see Table \ref{tab:final_features}). We choose to focus only on the Gini-ranked features rather than the SHAP-ranked features discussed in Section 4.2.1, because we are checking specifically for distributional shifts in the direct inputs that the SVM classifier receives from the Random Forest feature selection stage. The SHAP analysis, in contrast, explains the output of the entire integrated predictive pipeline. The choice serves as a more direct visual comparison of the distributions of the features that the core SVM model was trained on, compared to the feature distributions it encounters in the hold-out period.

    \item \textbf{Results}: The results of this analysis are shown in Figure \ref{fig:covariate_shift}. The distributions for the features plotted overlap closely. The lack of significant covariate shift provides strong evidence that the statistical properties of the important predictors did not change qualitatively over the hold-out period. It therefore supports the hypothesis that the model is operating within a similar data regime, which provides further credibility of the test set performance.

\end{itemize}

\begin{figure}[tbp]
\centering
\includegraphics[width=0.65\textwidth]{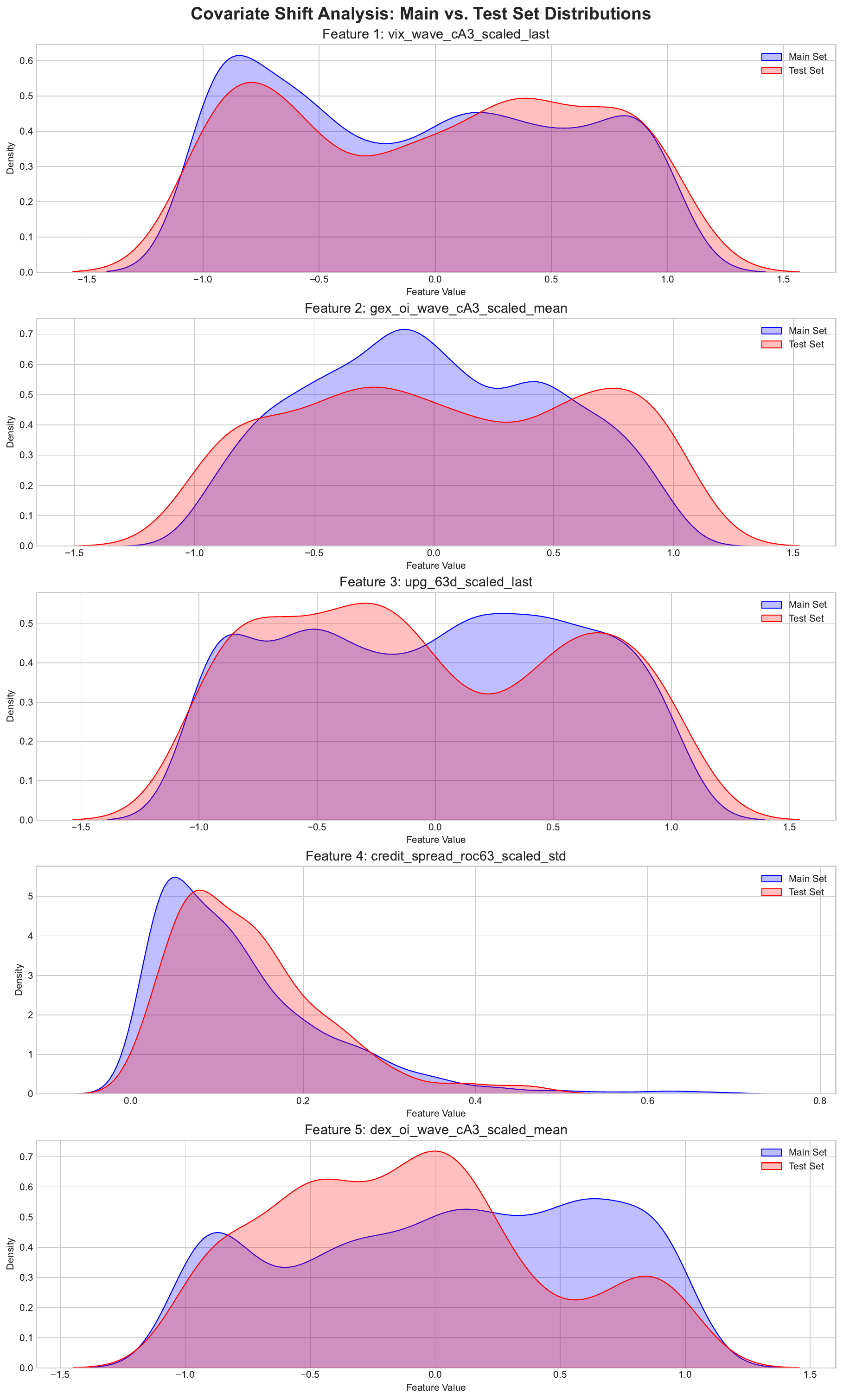}
\caption{Covariate Shift Analysis for Top Predictive Features}
\label{fig:covariate_shift}
\justify
\small{\textit{Notes:} This figure visually inspects for covariate shift by comparing the distributions of key input features between the main dataset and the hold-out test set. The plots display kernel density estimates (KDEs). The "Main Set" (blue) comprises the training and validation data from April 2013 to June 2023. The "Test Set" (orange) is the hold-out sample from July 2023 to June 2025. The five features shown are the top five predictors from the final 15-feature set, ranked by Gini importance from the Random Forest selector. The complete ranked list is available in Table \ref{tab:final_features}. The high degree of overlap between the distributions suggests the absence of significant covariate shift.}
\end{figure}

\subsection{Model Interpretation Stability: Concept Drift Analysis}

\begin{itemize}

    \item \textbf{Rationale}: The most subtle and important type of structural break is concept drift, when the distributional relationship between the features and the outcome has fundamentally changed. For example, an indicator that is previously important in nowcasting market capitulation is now insignificant for a different time period. We test this by looking at the stability of the model's own interpretation of feature importance over time.

    \item \textbf{Implementation}: We perform a SHAP stability analysis. We split the hold-out test set chronologically into 2 halves, and assess global SHAP feature importance bar plots for the first half and for the second half of the test dataset. Any considerable change in the rank or magnitude of importance of features, between these 2 plots, would signal concept drift.

    \item \textbf{Results}: As seen in Figure \ref{fig:shap_stability}, the SHAP importance plots are very much in line with each other across both halves of the hold-out. The highest ranked features in the first half (Panel (a)) remained the highest ranked features in the second half (Panel (b)), and their contributions are similar. In fact, the most SHAP-significant feature, \seqsplit{gex\_oi\_roc63\_scaled\_std}, has its mean SHAP values virtually identical. This stability provides strong evidence that the economic relationships underlying the model learned during training remained valid throughout the hold out period; the model did not need to "re-learn" what drives market troughs, and thus demonstrated its robustness against concept drift.

\end{itemize}

\begin{figure}[tbp]
\centering
\subfigure[First half of test set.]{
    \includegraphics[width=0.45\textwidth]{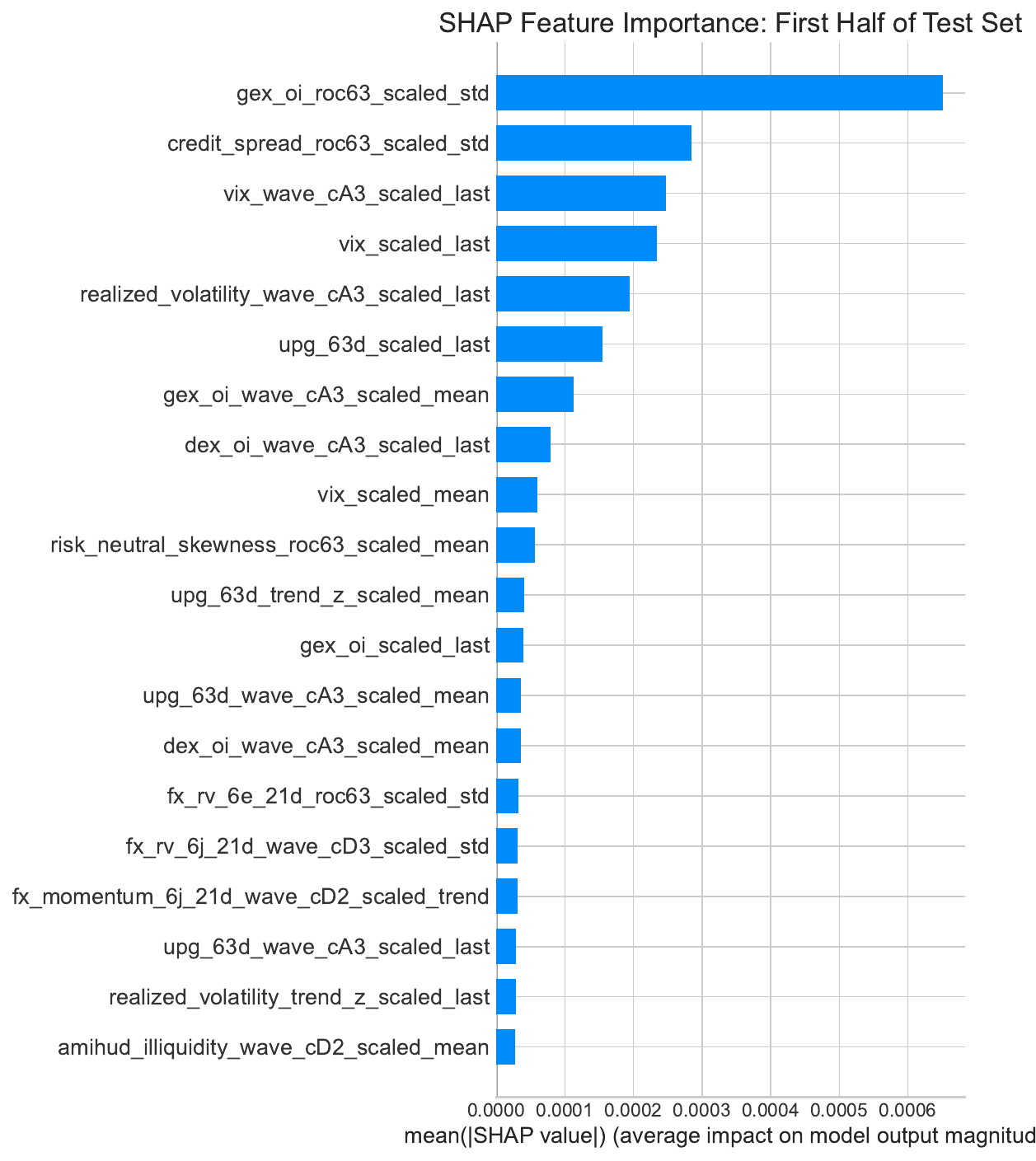}
    \label{fig:shap_stab1}
}
\subfigure[Second half of test set.]{
    \includegraphics[width=0.45\textwidth]{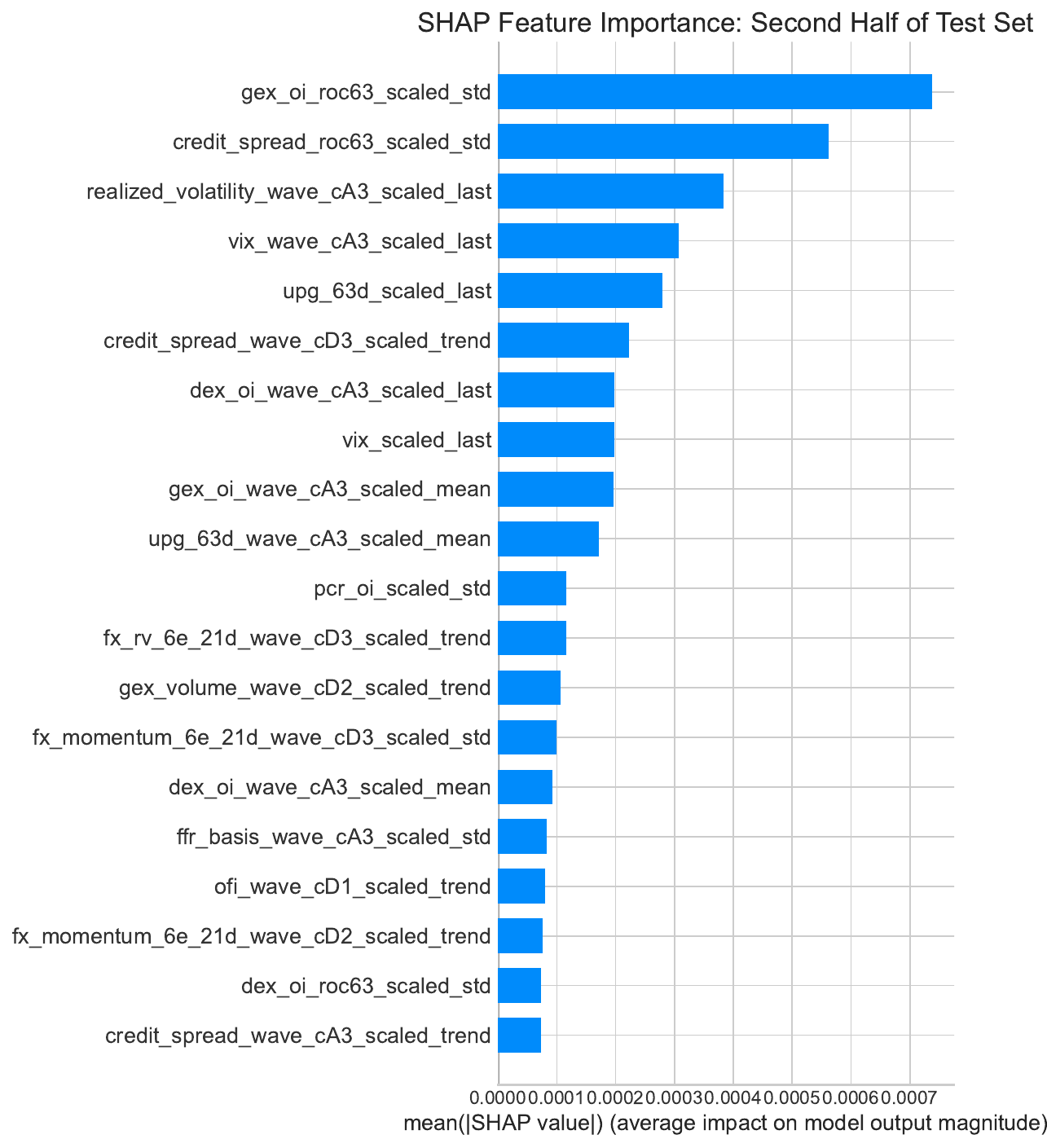}
    \label{fig:shap_stab2}
}
\caption{Stability of SHAP Feature Importance on the Hold-Out Test Set}
\label{fig:shap_stability}
\justify
\small{\textit{Notes:} This figure assesses the stability of the model's feature interpretations over time to test for concept drift. The panels display the mean absolute SHAP values for the top 20 features from the primary SVM model, calculated independently for two chronological sub-periods of the hold-out test set (July 2023 - June 2025). Panel (a) covers the first half (July 2023 - June 2024), and Panel (b) covers the second half (July 2024 - June 2025). The x-axis, `mean(|SHAP value|)`, quantifies the average magnitude of a feature's impact on the model's prediction. The high degree of consistency in the feature rankings and their relative magnitudes between the two periods indicates that the model's learned relationships are stable and robust against concept drift.}
\end{figure}

\section{Economic Importance and Signal Characteristics}

\label{sec:econ_sig}

While the statistics of Section 4 establish the predictive validity of the model, an important test is whether its forecasts have economically important signals that are robust to specific parameter choices. To go beyond a single point nowcast of market trough probability and investigate the economic properties of the signal provided by our model, we conduct a stylized backtest simulation on the hold-out test set, including a sensitivity analysis of the strategy holding period. We have no intention to propose a final, production-ready trading strategy. We use this backtest to diagnose and characterize the nature, strengths and weaknesses of the signal generated by the market capitulation nowcast model.

\subsection{Backtesting Methodology}

We simulate a simple trading strategy from the model's daily, out-of-sample calibrated trough probabilities. We perform our backtest using the E-mini S\&P 500 (ticker: ES), the most liquid equity index futures contract.

The actual simulations have the following rules:

\begin{enumerate}

    \item \textbf{Signal Generation}: A long position signal is created at the end of any day t where the calibrated trough probability exceeds the threshold of 5\% i.e., $P(\text{Trough})_t > 0.05$.

    \item \textbf{Trade Entry}: The long position is entered at the closing price of the ES contract on day t. All trades are size per contract, with a \$50 multiplier valued for each point movement.

    \item \textbf{Holding Period Sensitivity}: To test robustness, each position is held for periods of time from 5 to 20 trading days. We assess the performance across this spectrum.

    \item \textbf{Transaction Costs}: In order to take into account market frictions, a round-trip commission and slippage cost of \$5.00 is deducted from the profit or loss of each contract traded.

\end{enumerate}

To explore the performance characteristics of the model signals under different leverage and position-sizing rules, we evaluate two separate cases:

\begin{itemize}

    \item \textbf{Fixed-Size Strategy}: This is the baseline strategy, which enters one contract for each new signal, so that we can measure the economic value of the raw signal in the most direct way possible.

    \item \textbf{Pyramiding Strategy}: This approach is implemented to test the hypothesis that the model-based signals cluster at true reversals. It places position of $N$ size on the $N^{th}$ consecutive day the signal is active. This aggressively lever our position when the model has shown sustained conviction.

\end{itemize}

\subsection{Empirical Results and Interpretation}

The sensitivity analysis presented in Table \ref{tab:backtest_sensitivity_summary} provides a nuanced account of the economic value of the model. The results illustrate that the canonical signal is robust and economically valuable, while the pyramiding leverage presents limitation, providing a deep diagnosis of the nature of the signals it generate.

\begin{table}[htbp]
\centering
\caption{Economic Significance: Holding Period Sensitivity Analysis}
\label{tab:backtest_sensitivity_summary}
\resizebox{\textwidth}{!}{%
\begin{tabular}{llrrrrr}
\toprule
\textbf{Holding Period} & \textbf{Strategy} & \textbf{Total Net P\&L} & \textbf{Sharpe Ratio (Ann.)} & \textbf{Profit Factor} & \textbf{Max Drawdown} & \textbf{Max Drawdown (\%)} \\
\midrule
\multirow{2}{*}{5 Days (Baseline)} & Fixed-Size & \$31,247.50 & 0.38 & 1.22 & (\$52,682.50) & 55.66\% \\
& Pyramiding & \$797,222.50 & 1.62 & 2.77 & (\$176,712.50) & \textbf{186.71\%} \\
\midrule
\multirow{2}{*}{7 Days} & Fixed-Size & \$112,385.00 & 1.23 & 1.93 & (\$39,287.50) & 41.00\% \\
& Pyramiding & \$1,180,760.00 & 2.00 & 3.95 & (\$135,325.00) & \textbf{141.21\%} \\
\midrule
\multirow{2}{*}{10 Days} & Fixed-Size & \$200,985.00 & 2.01 & 3.00 & (\$25,000.00) & 10.76\% \\
& Pyramiding & \$1,404,622.50 & 2.18 & 4.42 & (\$239,230.00) & 15.74\% \\
\midrule
\multirow{2}{*}{12 Days} & Fixed-Size & \$235,210.00 & 2.03 & 3.34 & (\$56,052.50) & 18.37\% \\
& Pyramiding & \$1,165,810.00 & 1.21 & 2.50 & (\$694,205.00) & 40.40\% \\
\midrule
\multirow{2}{*}{20 Days} & Fixed-Size & \$217,385.00 & 1.23 & 1.95 & (\$229,240.00) & 57.59\% \\
& Pyramiding & \$735,522.50 & 0.63 & 1.45 & (\$1,634,867.50) & \textbf{80.06\%} \\
\bottomrule
\end{tabular}
}
\justify
\small{\textit{Notes:} This table summarizes the performance of two stylized trading strategies on the hold-out test set, evaluated across different holding periods. Both strategies trade E-mini S\&P 500 futures based on the model's out-of-sample trough probability forecasts. The "Fixed-Size" strategy trades one contract per signal. The "Pyramiding" strategy increases position size with each consecutive signal day. "Max Drawdown (\%)" exceeding 100\% (highlighted in bold) indicates a "risk of ruin" event, signifying a total loss of initial capital plus all accumulated profits. Complete trade logs for the 5-day baseline are available in \ref{sec:appendix_backtest_logs}.}
\end{table}

First, the performance on the Fixed-size strategy illustrates the strong economic edge of the raw signal. From the sensitivity analysis, there is a clear "sweet spot" for performance, as the annualized Sharpe Ratio peaks 2.01--2.03 for holding periods of 10 to 12 days. This is a strong finding because it indicates that the model's predictive power isn't simply some artifact of the first 5-day parameter selection, but pinpoints an actual market dynamic that plays out over a two- to three-week window following a capitulation signal.

Second, the Pyramiding strategy provides deep but cautionary insights. While the headline metrics look impressive, with a Sharpe Ratio of 2.18 at 10 days, they're completely swamped by the maximum drawdown numbers. As shown in Table \ref{tab:backtest_sensitivity_summary}, a maximum drawdown over 100\%, witnessed in the 5-day and 7-day periods, shows the "risk of ruin" event with a total loss of initial capital and all profits. Even the very large maximum drawdown for the other periods (e.g., 40.40\% at 12 days and 80.06\% at 20 days)—would represent a calamity for any type of real-world investment strategy. That makes the leveraged mechanical strategy completely uninvestable for its current design.

The drawdown failure does not suggest the model is problematic, as the backtest acts as a compelling diagnostic tool. The results show that our model is a good capitulation detector and a poor bear-to-bull trend-switching validator. The model detects moments of extreme panic that can lead to a sharp, V-shaped reversal, where the pyramiding strategy magnifies the return remarkably. However, it cannot consistently distinguish a prolonged market bottom from a "bear market rally," within a longer-duration downtrend. If the model falsely indicates a bottom in time that a market cannot recover from, the pyramiding logic creates a dangerously oversized position that loses its value when the market rolls again to a new low.

The economic significance of our model does not live simply in a straightforward trading rule, but in its predictive capability as a panic and capitulation detector. The sensitivity analysis has evaluated the model's signal robustness while effectively exposing the model's predictable failure mode. This understanding is essential because it implies for practical use cases, the model signal must not be used in isolation, but as a major signal component in the overall risk management analysis, likely in conjunction with additional longer-term regime filters to prevent prematurely market entries in prolonged downtrend.

\section{A Comparative Causal Analysis of Predictive Drivers}

\label{sec:causalanalysis}

Because our trough labels by the BB algorithm are retrospective, we must be cautious when interpreting the causal parameter $\theta$. $\theta$ does not represent a causal effect of a treatment on a future price path, but the effect on the contemporaneous state of the market, or more accurately, the effect on the probability that the market is in a state that would ultimately be determined as a trough in the future. This is an important distinction in order to interpret the policy implications of our findings (see Section 7.5).

Even though the predictive model in Section 4 establish strong out-of-sample nowcasting ability, the model does not give any causal analysis between the chosen features we identified and the market troughs. To advance towards robust causal interpretation rather than just statistical correlation, we implement the Double/Debiased Machine Learning (DML) framework. We conduct our analysis in two steps. As a foundational step, we first conduct and estimate the DML based Partially Linear Model (DML-PLR), which is a standard approach in the literature. Acknowledging its limitation, we conduct a more flexible and appropriate DML specification by estimating the Average Partial Effect (APE). APE accounts for binary outcomes, feature interactions, and non-linear treatment effects. With this comparative approach, we can identify robust causal drivers and demonstrate how model specification can impact our economic conclusions.

\subsection{Baseline Model: DML of the Partially Linear Model (DML-PLR)}

Our causal analysis starts with a base line: the DML approach for Partially Linear Regression (PLR) models, as outlined by \citet{Chernozhukov2018}. This specification provides a point of reference, with the assumption that the treatment effect is constant and additively separable. We write the structural form as:

$$\mathbf{Y} = \theta \mathbf{D} + g(\mathbf{X}) + \epsilon$$

where $\mathbf{Y}$ is the trough outcome, $\mathbf{D}$ is the treatment variable (a single indicator of interest), $\mathbf{X}$ is a high-dimensional vector of all other features that are potentially confounding, and $g(\cdot)$ is an unknown nonlinear function. The DML2 algorithm with cross-fitting provides a $\sqrt{N}$-consistent and asymptotically normal estimate for the constant treatment effect $\theta$ by flexibly modeling two nuisance functions: the outcome model $\hat{l}_0 (\mathbf{X}) = \mathbb{E}[\mathbf{Y} | \mathbf{X}]$, and the treatment model $\hat{m}_0 (\mathbf{X}) = \mathbb{E}[\mathbf{D} | \mathbf{X}]$.

In order to avoid the arbitrary choice of a single machine learning model for the nuisance functions, we use a data-driven selection process in each cross-fitting fold. In estimating the conditional mean of the treatment, $\hat{m}_0(\mathbf{X})$, we conduct a 'horse race'. A `GradientBoostingRegressor` and a `LassoCV` model are trained on the training portion of the fold. The model that exhibits better predictive performance, as assessed by out-of-sample R-squared on the validation portion of the fold, is dynamically selected for the predictions. This automatic selection improves the robustness of the DML procedure.

While the PLR specification is a standard benchmark, it has two important limitations for this setting. First, for our binary outcome $\mathbf{Y} \in \{0, 1\}$, PLR is a Linear Probability Model (LPM), which could generate predictive probabilities outside the logical [0, 1] range. Second, it imposes the restrictive assumption that the causal effect of the treatment $\theta$ is constant and additively separable from the effects of confounders. That assumption is not likely to hold in financial markets, where the signaling power of an indicator often depends on the market context.

\subsection{Main Model: DML for the Average Partial Effect (DML-APE)}

To circumvent the limitations of the PLR framework, we employ a more flexible DML estimator based on an interactive model. We define the conditional probability of a trough as a non-linear interactive function:

$$P(\textbf{Y}=1 | \textbf{D}=d, \textbf{X}=x) = l(d,x)$$

This specification is theoretically valid for a binary outcome and enables the treatment effect to vary with the state of the high dimensional confounders. The causal parameter that we are interested in is the Average Partial Effect (APE), $\theta_0$, defined as the expected gradient of the conditional probability function with respect to the treatment:

$$ \theta_0 = \mathbb{E}_{\textbf{D},\textbf{X}}\left[\frac{\partial l(\textbf{D},\textbf{X})}{\partial \textbf{D}}\right]$$

APE measures the average change in probability of a market trough for a one unit increase in treatment, averaged across the entire data distribution. In order to estimate APE reliably, this framework requires learning these three nuisance functions:

\begin{enumerate}

    \item The outcome model (classification): $l(d,x)=\mathbb{E}[\textbf{Y}|\textbf{D}=d, \textbf{X}=x]$.

    \item The treatment mean model (regression): $m(x)=\mathbb{E}[\textbf{D}|\textbf{X}=x]$. 

    \item The treatment conditional variance model (regression): $v(x) = \mathbb{E}[(\mathbf{D} - m(\mathbf{X}))^2 | \mathbf{X}=x]$.

\end{enumerate}

Similar to the PLR approach, we run both `GradientBoostingRegressor` and `LassoCV` models in each cross-fitting fold as a "horse race" to determine the best performing estimator of the conditional mean $\hat{m}_0(\mathbf{X})$ and conditional variance $\hat{v}_0(\mathbf{X})$, based on out-of-sample R-squared. It ensure that nuisance parameters are estimated from the most appropriate functional form for that slice of data.

Given a standard and flexible assumption that the treatment is a heteroskedastic Gaussian process conditional on confounders, the Neyman-orthogonal score function for APE is:

$$\psi(\mathbf{W}; \theta, \eta) = \underbrace{\frac{\partial l(\mathbf{D}, \mathbf{X})}{\partial \mathbf{D}} - \theta}_{\text{Naive Score}} + \underbrace{\frac{\mathbf{D} - m(\mathbf{X})}{v(\mathbf{X})} \left( \mathbf{Y} - l(\mathbf{D}, \mathbf{X}) \right)}_{\text{Bias Correction}}$$

where $\mathbf{W} = (\mathbf{Y}, \mathbf{D}, \mathbf{X})$ and $\eta = (l, m, v)$. A complete derivation of this score function is provided in \ref{sec:appendixapederiv}. For practical estimation of the score we need to compute each of its constituents. The partial derivative term $\partial l(\mathbf{D}, \mathbf{X}) / \partial \mathbf{D}$ is numerically computed with a standard finite difference approach in the fitted outcome model $\hat{l}$. The bias correction term, which contains the three nuisance functions, is important because it ensures the final estimate of $\theta$ is robust to first order estimation error in the machine learning models. In order to be robust against the estimation "noise" from the nuisance models, especially possible outliers when $\hat{v}(\mathbf{X})$ is close to zero, our final point estimate $\hat{\theta}$ is the median of the scores computed on the out-of-sample fold, and inference proceeds with the non-parametric bootstrap of those scores. A full treatment of justification for this approach can be found in \ref{sec:appendix_ape_median}. The nuisance functions are estimated by the same horse race approach with GradientBoostingRegressor and LassoCV learners as specified in the DML-PLR analysis (section 7.1).

\subsection{Model Specification and Endogeneity}

Any credible causal estimate depend on good model specification in order to limit endogeneity. To keep comparisons between DML-PLR and DML-APE fair, the methods adopted to mitigate bad controls and sensitivity to unobserved confounding are applied equally to both framework, as explained below.

\subsubsection{Bad Controls and Multicollinearity}

In order to eliminate spurious results from multicollinearity or "bad controls", we have an explicit exclusion map. When an aggregated variable (e.g. \seqsplit{vrp\_scaled\_mean}) is selected to be the treatment variable $\mathbf{D}$; we exclude all other aggregated features with the same parent indicator (i.e. \seqsplit{vrp\_scaled\_std}, \seqsplit{vrp\_scaled\_trend} from the set of potential confounders $\mathbf{X}$. We also used an exclusion map to remove any features that would be mechanistic components of the treatment. For example, because Variance Risk Premium (VRP) is  defined by VIX and Realized Volatility (RV), we eliminated all features based on VIX or RV from $\mathbf{X}$ when VRP related features are the treatment $\mathbf{D}$. This procedure is essential to estimate the total causal effect of VRP, rather than an effect partially-out by its own constituent parts, which would be misleading.

\subsubsection{Sensitivity to Unobserved Confounders}

Even though the DML framework accounts for the observed confounders in $\mathbf{X}$, its estimate could be biased by unobserved confounders. To address this, we subject all statistical significant DML estimates to a formal sensitivity analysis based on \citet{Cinelli2020}. It quantifies how large an unobserved confounder must be (expressed by its partial $R^2$ with the treatment and the outcome) to reject the causal claim. The sensitivity analysis enable formal elimination of hypotheses that are plausible under DML but not robust under unobserved confounding. Note that while our DML estimators are non-parametric, the sensitivity analysis framework is based on a linear model. We use it as a pragmatic and conservative guide to validate our causal claim against unobserved "worst-case" linear confounders.

\subsection{Causal Effect Estimates: A Comparative Analysis}

\label{sec:causalestimates}

Our comparative causal analysis shows that robust economic insight depends on model specification. While the DML-PLR framework provides a baseline, its linear assumption masks complex non-linear and interactive relationships, sometimes even misinterpreting causal effect. On the other hand, the DML-APE framework, powered by its flexibility and non-linear interactions, yields a richer and more plausible set of drivers, and in some cases reversing the sign of the estimates from the DML-PLR framework. Table \ref{tab:dml_comparison} displayed an organized comparison of the results of both frameworks. The full 27 robust estimates from the DML-PLR model and 48 robust estimates for the DML-APE model are listed in \ref{sec:appendix_results}. Overall, the comparisons yield three core insights.

First, a small group of core causal drivers are robust to model specification. For example, both models establish that trend in the Fed Funds futures slope (\seqsplit{ffr\_slope\_scaled\_trend}) has a statistically significant, negative causal effect on trough probability. The agreement affirms the hypothesis that market perceptions of future monetary easing establish an important stabilizing force, whether the force is linear or non-linear.

Second, the shift to a more flexible APE framework allows elimination of findings that might be spurious due to the linearity assumption in PLR. The standard deviation of the credit spread (\seqsplit{credit\_spread\_scaled\_std}) is a clear example. The PLR model estimated a significant negative coefficient on this variable, suggesting that spread volatility is stabilizing. However, its effect is statistically indistinguishable against zero in the APE specification. It shows that the negative PLR model estimate is likely an artifact of linearity, and APE successfully captures how credit spread volatility interacts with wider market context.

Third, and most importantly, the main contribution of the DML-APE model is its identification of new causal pathways and their plausible economic implications. The APE model discovers that the volatility of measures of options based risk appetite, (e.g., \seqsplit{gex\_oi\_trend\_z\_scaled\_std}, \seqsplit{risk\_neutral\_skewness\_scaled\_std}, and the VRP measure \seqsplit{vrp\_roc63\_scaled\_std}) are causal drivers of troughs, a group of causal drivers missed by the PLR framework. This indicates a more sophisticated market mechanism, where it is not only the level of fear, but also its rate of change and persistence, that causally drives market capitulation.

The DML-APE framework also reverses the sign of several causal estimates, which resolves counter-intuitives economic interpretation from the linear PLR. The volatility of the Amihud illiquidity trend (\seqsplit{amihud\_illiquidity\_trend\_z\_scaled\_std}) is a prominent example. The PLR model inferences a negative effect ($\hat{\theta}=-0.0608$), as shown in Table \ref{tab:dml_plr_full_results} suggests that volatility in the illiquidity trend is a stabilizing driver. The APE model then reverses the sign of the estimate, producing a robust positive effect ($\hat{\theta}=0.0160$ in Table \ref{tab:dml_ape_full_results}). This illustrates the risk of overly restrictive assumption made by the PLR model. The APE model correctly captures the market mechanism that, on average, rising instability in market liquidity is a causal precursor to a trough. Similarly, the effect of the volatility of Put / Call ratio (\seqsplit{pcr\_oi\_roc63\_scaled\_std}) flips from negative in PLR to the more intuitive positive in the APE framework, aligning with the economic intuition that a volatile and increasing demand for puts is a destabilizing force for the market.

To sum up, the comparative analysis supports the choice of DML-APE as the primary causal framework. We have moved beyond unrealistic linear assumptions to reveal that the causal drivers of market troughs are rooted firmly in the non-linear dynamic and interactions of market volatility, options implied risk appetite, and liquidity.

\begin{table}[tbp]
\centering
\caption{Comparative DML Causal Estimates: PLR vs. APE Models}
\label{tab:dml_comparison}
\resizebox{\textwidth}{!}{%
\begin{tabular}{@{}llrrrc@{}}
\toprule
\textbf{Theme} & \textbf{Treatment Variable (D)} & \textbf{Model} & \textbf{Coeff. ($\hat{\theta}$)} & \textbf{$p$-value} & \textbf{Robust?} \\
\midrule
\multicolumn{6}{@{}l}{\textit{Finding 1: Consistent Negative Effect of Easing Expectations}} \\
Monetary Policy & \seqsplit{ffr\_slope\_scaled\_trend} & PLR & -0.1436 & 0.0010 & Yes \\
& & APE & -0.0073 & $<$0.0001 & Yes \\
\midrule
\multicolumn{6}{@{}l}{\textit{Finding 2: Effect of Credit Spread Volatility Lost Robustness}} \\
Credit Conditions & \seqsplit{credit\_spread\_scaled\_std} & PLR & -0.0524 & $<$0.0001 & Yes \\
& & APE & - & - & No \\
\midrule
\multicolumn{6}{@{}l}{\textit{Finding 3: New Volatility-Based Drivers Gained Robustness}} \\
Options Risk Appetite & \seqsplit{gex\_oi\_trend\_z\_scaled\_std} & PLR & - & - & No \\
& & APE & 0.0773 & $<$0.0001 & Yes \\
Volatility Risk Premium & \seqsplit{vrp\_roc63\_scaled\_std} & PLR & - & - & No \\
& & APE & -0.0021 & 0.0099 & Yes \\
\midrule
\multicolumn{6}{@{}l}{\textit{Finding 4: Causal Sign Reversal for Liquidity and Sentiment}} \\
Market Liquidity & \seqsplit{amihud\_illiquidity\_trend\_z\_scaled\_std} & PLR & -0.0608 & 0.0001 & Yes \\
& & APE & 0.0160 & $<$0.0001 & Yes \\
Market Sentiment & \seqsplit{pcr\_oi\_roc63\_scaled\_std} & PLR & -0.0549 & 0.0057 & Yes \\
& & APE & 0.0241 & $<$0.0001 & Yes \\
\bottomrule
\end{tabular}%
}
\justify
\small{\textit{Notes:} This table compares robust causal estimates for selected variables from the DML-PLR (Partially Linear Regression) and DML-APE (Average Partial Effect) models. Robustness is determined by a formal sensitivity analysis to unobserved confounding \`a la \citet{Cinelli2020}. A "Yes" indicates the finding is statistically significant ($p<0.05$) and passed the sensitivity check. A "No" indicates the finding is either not statistically significant or not robust. The comparison shows how moving to the more flexible APE model changes the set of identified causal drivers, highlighting findings that remain consistent, lose robustness, gain robustness, or reverse in sign. Full results are in \ref{sec:appendix_results}.}
\end{table}

\subsubsection{Hypotheses Discarded: the VIX Trend Example}

The formal DML and sensitivity analysis allows us to formally discard plausible but non-robust causal hypotheses. A prime example is the trend of VIX. In the initial DML procedure, \seqsplit{vix\_roc63\_scaled\_trend} produced a statistically significant coefficient. However, the nuisance model $R^2$ is greater than 0.7, meaning that much of its variation is already explained by other controls like credit spreads. The formal sensitivity analysis shows that the finding is highly sensitive to unobserved confounding. Therefore, we cannot make robust causal claim about the VIX trend and discard it from our primary findings. This filtering procedure by sensitivity analysis is important for credible economic interpretations in the presence of unobserved confounding and high dimensional confounders.

\subsection{Causal Interpretation through Structural Economic Models}

\label{sec:causalinterpretation}

Before bridging our causal estimates by DML with structural models, we need to revisit how we interpret causal estimates. Because our estimand represents the effect on the contemporaneously measured market state, as opposed to the future market outcome, its interpretation and policy implications are nuanced. An intervention on a causal driver of market trough should not be interpreted as a tool to prevent future capitulation, but as a way to cure the current market dynamic underlying the market capitulation. Even though an effective policy that causally reduce market illiquidity might not halt a bear market, it can fix the "fire-sale" phenomenon of the market trough. This perspective reminds us that our findings identify causal triggers for the characteristics of a market bottom, not necessarily the trough event itself. 

The DML framework provides strong, reduced-form evidence about the high-frequency causal drivers of market troughs, but the broad question is how these short term treatment effects, on a timescale of a few days, bridge with canonical structural macro-finance models of financial crisis, which describes the accumulation of financial market vulnerabilities over the long term. We frame this connection as "state and trigger". Structural models like \citet{Bernanke1999}, \citet{Kiyotaki1997}, \citet{He2013} describe the gradual buildup of systemic fragility using state variables such as intermediary capital constraints or aggregate balance sheet health. Our empirical DML framework then uncovers the acute, observable market signals that "trigger" market trough in the fragile state. In this section, we utilize the comparative analyses between the baseline DML-PLR and DML-APE framework to achieve a more granular, structural-macrofinance model-based understanding of these triggers.

\subsubsection{The Financial Accelerator: Credit Conditions and Policy Channels}

Our DML model provides a high-frequency snapshot of the financial accelerator framework of \citet{Bernanke1999}. The main mechanism is the external finance premium, which is the wedge between the cost of external finance and risk free rate. It rises following the deterioration of a borrower's balance sheet. The model expresses this with an equilibrium condition:

$$\frac{\mathbb{E}_t[\mathbf{R}_{t+1}^k]}{R_{t+1}} = s\left(\frac{\mathbf{N}_{t+1}}{\mathbf{Q}_t \mathbf{K}_{t+1}}\right), \quad \text{with} \quad s'(\cdot) < 0$$

where the left-side of the equation represents the ratio of expected return to a firm’s capital ($\mathbb{E}_t[\mathbf{R}_{t+1}^k]$) to the risk-free rate ($R_{t+1}$), and the right-side represents the external finance premium. The premium is the function $s(\cdot)$ that relies on the ratio of the firm's net worth ($\mathbf{N}_{t+1}$) to the value of its capital assets ($\mathbf{Q}_{t} \mathbf{K}_{t+1}$), which resembles the firm's underlying balance sheet condition. The important property is $s'(\cdot) < 0$: as a firm's net worth declines relative to its asset value, the firm's balance-sheet degrades and the external finance premium $s(\cdot)$ increases. Our \seqsplit{credit\_spread} indicator is a market-based measure of this external finance premium.

The comparison between DML-PLR and DML-APE reveal a subtle but important difference. The DML-PLR model finds a strong, negative causal effect of the standard deviation of the credit spread (\seqsplit{credit\_spread\_scaled\_std}, $\hat{\theta}=-0.0524$), while that strong effect disappears in the less restricted DML-APE model. This suggests that the true relationship is non-linear; the causal effect of credit spread volatility is conditional on the market state, and the DML-APE model correctly captures this interaction that confounds the linear estimate in the DML-PLR model.

In contrast, the causal pathway via monetary policy expectations demonstrates strong coherence across both models. The trend in the Fed Funds futures slope (\seqsplit{ffr\_slope\_scaled\_trend}) has a robust and statistically significant negative causal effect on market capitulation for both DML-PLR ($\hat{\theta}=-0.1436$) and DML-APE ($\hat{\theta}=-0.0073$) models. This specification-invariant finding shows that the expectation of future policy easing has an immediate, causal stabilising influence, likely due to better expectation of future corporate net worth $\mathbb{E}_t[\mathbf{N}_{t+k}]$, which prevents an increase in the external finance premium.

\subsubsection{Leverage Cycles, Liquidity Spirals, and Fire Sales}

We have also empirically captured triggers for leverage cycles \citep{Geanakoplos2010} and liquidity spirals \citep{Brunnermeier2009}. In their structural frameworks, a damaging "scary news" can cause collateral values to drop below regulatory thresholds, forcing leveraged agents into fire sales. The fire sales are further aggravated by a feedback loop between market illiquidity and tightening funding conditions. The mechanics of a fire sale are described by the liquidity spiral in Brunnermeier's model:

$$|\Lambda^{j}_{1}|=m^{j}_{1}(\phi_{1}-1)$$

where $|\Lambda^{j}_{1}|$ represents market illiquidity—specifically how much the price moves for an order size. The illiquidity is the product of two mutually-reinforcing mechanisms: tightening margin requirements $m^{j}_{1}$, which shows a leveraged collapse, and funding illiquidity $\phi_{1}-1$, which demonstrates the difficulty of financing asset purchases when there is a fire sale. The equation is a vicious cycle, as the tightening of the margin forces sales, which are amplified by funding illiquidity, which then tightens market liquidity, creating more margin calls.

Our DML comparative analysis shows the empirical proxies for the triggers of the theoretical structural models. The most significant evidence is the causal channel for market liquidity. For the volatility of the Amihud illiquidity trend (\seqsplit{amihud\_illiquidity\_trend\_z\_scaled\_std}), which is a direct proxy for the instability of the market illiquidity $|\Lambda^{j}_{1}|$, the DML-PLR framework produces a robust but counter-intuitive negative coefficient ($\hat{\theta}=-0.0608$). The DML-APE framework then reversed the sign and found a robust and theoretically consistent positive coefficient ($\hat{\theta}=0.0160$). This sign flip is critical, as the APE model's finding aligns with the structural theory: as instability in market iliquidity increases, it is a causal precursor to a trough, confirmation the illiquidity spiral market dynamic.

\subsubsection{A Unifying Paradigm: Intermediary Asset Pricing}

We can unify these findings under the modern structural framework of intermediary asset pricing \citep{He2013}. The framework constructs that the financial system's risk-bearing capacity comes from the wealth of the economy's specialized, leveraged intermediary sector. If this sector is financially constrained, the market price for risk rises non-linearly. The relationship is expressed by the intermediary pricing kernel:

$$\mathbb{E}_t[d\mathbf{R}_t] - r_t dt = \alpha_t^I \text{Var}_t[d\mathbf{R}_t]$$

Here, the market's excess return on average (the equity risk premium on the left) is equal to the market variance times the intermediary sector's risk aversion ($\alpha_t^I$). The key is that $\alpha_t^I$ is not constant; it increase non-linearly as intermediaries' capital is depleted. A market capitulation is the outcome of an economy being entrenched in this period of financially constraint, with a high and volatile $\alpha_t^I$. Our comparative DML analysis provides a rich description of the trigger of troughs in this region.

\begin{itemize}

    \item \textbf{The Empirical Signature of a Constrained Regime.} The results of the DML-APE model demonstrates that the intermediary sector is constrained, with a high and unstable $\alpha_t^I$, around market capitulation. The causal effect of the volatility of options-implied risk measures, specifically \seqsplit{risk\_neutral\_skewness\_scaled\_std} ($\hat{\theta}=0.0359$) and \seqsplit{risk\_neutral\_kurtosis\_scaled\_std} ($\hat{\theta}=0.0957$), corroborates the non-linear, erratic market price of risk expected when constrained intermediaries cannot absorb shocks smoothly.

    \item \textbf{Binding Constraints, Deteriorating Market Liquidity.} A high $\alpha_t^I$ implies intermediaries withdraw due to risk aversion, damaging market liquidity. The sign switch for the Amihud illiquidity volatility (\seqsplit{amihud\_illiquidity\_trend\_z\_scaled\_std}) from negative in the DML-PLR model to the positive in DML-APE model, ($\hat{\theta}=0.0160$), is the empirical smoking gun. The APE model correctly identifies that the increasing instability in market illiquidity, directly driven by the intermediaries withdrawing capital, is a robust and significant causal driver of the market trough.

    \item \textbf{The Policy Stabilization Channel.} The robust, stabilizing causal channel of the Fed Funds futures slope (\seqsplit{ffr\_slope\_scaled\_trend}, with the APE estimate $\hat{\theta}=-0.0073$), fits perfectly in the structural model. When the market expects a dovish policy, such as a cut in the Fed Funds rate, the intermediaries' future funding conditions and franchise value are improved. That eases their capital constraints now, enhances the liquidity and risk-bearing of the financial market, and reduces their effective risk aversion $\alpha_t^I$, moving the market away from the capitulation region.

\end{itemize}

To sum up, our nowcasting model's successful prediction capabilities stem from its ability to learn the short-time empirical signatures of the core triggers and mechanisms, corroborated by the long-term structural macro-finance theories on financial crises. Our comparative causal analysis reveals that, to correctly understand the triggers for market troughs, the causal frameworks require specifications that account for non-linearity and interaction. DML-APE offers a more theoretically coherent account of how the long-term latent risks, delineated by the structural macro-finance models, manifest short-term observable market capitulation, providing a rich empirical validation of modern asset pricing theory.

\section{Conclusions}

In this paper, we tackled the market trough prediction challenge by a dual-track approach. First, we established a high-performance machine learning pipeline, culminating in a SVM classifier, which produced an ROC AUC of 0.89 out-of-sample, and whose predictions we rendered interpretable via SHAP analysis. A stylized backtest of a simple E-mini S\&P 500 futures strategy confirmed the economic significance of the nowcast signal as a market capitulation detector, which performed well in its ability to detect V-shaped price reversals, though it struggles in grinding bear markets. Second, and more importantly, we moved beyond statistical correlation and did a comparative causal analysis. We argued that while the DML-PLR model provided a good baseline, the more flexible DML-APE specification properly accounted for the binary outcome and non-linear interactions problem and provided more credible economic inferences that aligned with structural macrofinance theories. The DML-APE framework not only corrected the sign of some key causal effects, but also discovered new causal pathways related to market volatility.

Our primary causal findings, robust and statistically significant under DML-APE specification and formal sensitivity analysis, suggest that the triggers for market capitulations are fundamentally grounded in the non-linear dynamics of market instability. We found that it is not simply the level of fear, but the volatility of options-implied risk appetite (e.g., GEX, risk-neutral skewness) and the instability of market illiquidity (Amihud illiquidity) that causally drive market troughs. Our results provide high-frequency empirical validation for structural long term intermediary asset pricing theories, which conceptualizes market troughs as representing a phase transition into a constrained, non-linear regime whereby the market price of risk is erratic.

This research provides several channels for future work. Firstly, an extension of this research is to tackle the symmetric but distinct challenge of predicting market peaks. It likely requires a very different feature set to capture market euphoria before the peaks compared to the environment of capitulation. Secondly, a more extended dataset going back to pre-2013, including the 2008 global financial crisis (GFC),  would be an important out-of-sample test of the model's performance and robustness across distinct macroeconomic regimes, though it is hard to obtain high frequency option market data before 2013. Methodologically, while our framework enhances DML-PLR by estimating the Average Partial Effects, future research could extend to Conditional Average Partial Effects (CAPE), which would provide further insights on how these causal impacts differ contingent on the market states. Lastly, a promising frontier is to try to use Physics-Informed Machine Learning (PIML) to connect high frequency causal trigger of market trough revealed by DML-APE with long-term market dynamics explained by established structural macrofinance model through Hamilton–Jacobi–Bellman (HJB) equation, and then use deep reinforcement learning to make the framework tractable and create a generalized causal discovery for financial capitulation

\section*{Acknowledgements}

\clearpage
% 'aer' style is a close match for the RFS author-year format.
\bibliographystyle{aer}
\bibliography{riaibib} % <-- This should point to your .bib file

\clearpage
\appendix
\section{Full Bry-Boschan Algorithm Implementation}
\label{sec:appendix_bb_algo}

This appendix provides the full, unabridged pseudocode for the modified Bry-Boschan algorithm used to identify market turning points, as referenced in Section 2.2.

\captionof{algorithm}{Full Modified Bry-Boschan Algorithm}
\label{alg:bry_boschan_full}
\begin{algorithmic}[1]
\Statex \textbf{Require}: Log price series $P_t$, window \protect\seqsplit{order}, \protect\seqsplit{min\_phase}, \protect\seqsplit{min\_cycle}.
\Statex \textbf{Ensure}: A DataFrame \protect\seqsplit{turns} of significant peaks and troughs.
\Statex
\Procedure{IdentifyTurns}{$P_t$, \protect\seqsplit{order}, \protect\seqsplit{min\_phase}, \protect\seqsplit{min\_cycle}}
    \State Initialize \protect\seqsplit{turns} with all local peaks and troughs from $P_t$, sorted by date.
    \State \protect\seqsplit{turns} $\gets$ \Call{EnforceAlternation}{\protect\seqsplit{turns}}
    \State \protect\seqsplit{turns} $\gets$ \Call{CensorPhases}{\protect\seqsplit{turns}, \protect\seqsplit{min\_phase}}
    \State \protect\seqsplit{turns} $\gets$ \Call{CensorCycles}{\protect\seqsplit{turns}, \protect\seqsplit{min\_cycle}, $P_t$}
    \State \Return \protect\seqsplit{turns}
\EndProcedure
\Statex
\Function{EnforceAlternation}{\protect\seqsplit{turns}}
    \State Let \protect\seqsplit{processed\_turns} be an empty list.
    \For{each \protect\seqsplit{current\_turn} in \protect\seqsplit{turns}}
        \If{\protect\seqsplit{processed\_turns} is empty or \protect\seqsplit{current\_turn.type} $\neq$ \protect\seqsplit{last\_turn.type}}
            \State Add \protect\seqsplit{current\_turn} to \protect\seqsplit{processed\_turns}.
        \Else \Comment{If same type, keep the more extreme one.}
            \If{(\protect\seqsplit{current\_turn.type} is 'Peak' \textbf{and} \protect\seqsplit{current\_turn.value} > \protect\seqsplit{last\_turn.value}) \textbf{or} \\ \qquad (\protect\seqsplit{current\_turn.type} is 'Trough' \textbf{and} \protect\seqsplit{current\_turn.value} < \protect\seqsplit{last\_turn.value})}
                \State Replace last element of \protect\seqsplit{processed\_turns} with \protect\seqsplit{current\_turn}.
            \EndIf
        \EndIf
    \EndFor
    \State \Return \protect\seqsplit{processed\_turns}
\EndFunction
\Statex
\Function{CensorPhases}{\protect\seqsplit{turns}, \protect\seqsplit{min\_phase}}
    \Repeat
        \State \protect\seqsplit{has\_changed} $\gets$ false
        \For{\protect\seqsplit{i} from 0 to length of \protect\seqsplit{turns} - 2}
            \If{Duration(\protect\seqsplit{turns}[\protect\seqsplit{i}], \protect\seqsplit{turns}[\protect\seqsplit{i}+1]) < \protect\seqsplit{min\_phase}}
                \State Drop the less extreme of \protect\seqsplit{turns}[\protect\seqsplit{i}] and \protect\seqsplit{turns}[\protect\seqsplit{i}+1].
                \State \protect\seqsplit{turns} $\gets$ \Call{EnforceAlternation}{\protect\seqsplit{turns}}
                \State \protect\seqsplit{has\_changed} $\gets$ true; \textbf{break} \Comment{Restart scan.}
            \EndIf
        \EndFor
    \Until{\protect\seqsplit{has\_changed} is false}
    \State \Return \protect\seqsplit{turns}
\EndFunction
\Statex
\Function{CensorCycles}{\protect\seqsplit{turns}, \protect\seqsplit{min\_cycle}, $P_t$}
    \Repeat
        \State \protect\seqsplit{has\_changed} $\gets$ false
        \For{\protect\seqsplit{i} from 0 to length of \protect\seqsplit{turns} - 3}
            \If{Duration(\protect\seqsplit{turns}[\protect\seqsplit{i}], \protect\seqsplit{turns}[\protect\seqsplit{i}+2]) < \protect\seqsplit{min\_cycle}}
                \State Let (\protect\seqsplit{t1}, \protect\seqsplit{t2}, \protect\seqsplit{t3}) be (\protect\seqsplit{turns}[\protect\seqsplit{i}], \protect\seqsplit{turns}[\protect\seqsplit{i}+1], \protect\seqsplit{turns}[\protect\seqsplit{i}+2]).
                \If{\protect\seqsplit{t2} is the true extremum between \protect\seqsplit{t1} and \protect\seqsplit{t3}}
                     \State Drop \protect\seqsplit{t1} and \protect\seqsplit{t3}.
                \Else
                    \State Drop \protect\seqsplit{t2}.
                \EndIf
                \State \protect\seqsplit{turns} $\gets$ \Call{EnforceAlternation}{\protect\seqsplit{turns}}
                \State \protect\seqsplit{has\_changed} $\gets$ true; \textbf{break} \Comment{Restart scan.}
            \EndIf
        \EndFor
    \Until{\protect\seqsplit{has\_changed} is false}
    \State \Return \protect\seqsplit{turns}
\EndFunction
\end{algorithmic}

\section{Backtesting Details}

\subsection{Full Trade Logs for Baseline (5-Day Holding Period)}
\label{sec:appendix_backtest_logs}

This appendix provides the complete, unabridged trade logs for the two backtesting scenarios discussed in Section \ref{sec:econ_sig}. Table \ref{tab:backtest_log_fixed} details the 48 trades from the Fixed-Size strategy. Table \ref{tab:backtest_log_pyramid} details the trades from the Pyramiding strategy.

\begin{table}[htbp]
\centering
\caption{Full Trade Log: Fixed-Size Strategy}
\label{tab:backtest_log_fixed}
\resizebox{\textwidth}{!}{%
\begin{tabular}{rrrrr|rrrrr|rrrrr}
\toprule
\textbf{\#} & \textbf{Entry Date} & \textbf{Exit Date} & \textbf{Entry Price} & \textbf{P\&L (\$)} & \textbf{\#} & \textbf{Entry Date} & \textbf{Exit Date} & \textbf{Entry Price} & \textbf{P\&L (\$)} & \textbf{\#} & \textbf{Entry Date} & \textbf{Exit Date} & \textbf{Entry Price} & \textbf{P\&L (\$)} \\
\midrule
0 & 2023-10-13 & 2023-10-20 & 4351.50 & -5355.0 & 16 & 2024-08-02 & 2024-08-09 & 5358.00 & 345.0 & 32 & 2025-03-05 & 2025-03-12 & 5845.25 & -11242.5 \\
1 & 2023-10-16 & 2023-10-23 & 4402.25 & -7417.5 & 17 & 2024-08-05 & 2024-08-12 & 5273.75 & 4932.5 & 33 & 2025-03-06 & 2025-03-13 & 5764.75 & -8055.0 \\
2 & 2023-10-17 & 2023-10-24 & 4392.00 & -6192.5 & 18 & 2024-08-06 & 2024-08-13 & 5243.50 & 10707.5 & 34 & 2025-03-07 & 2025-03-14 & 5771.75 & -4542.5 \\
3 & 2023-10-18 & 2023-10-25 & 4344.00 & -7630.0 & 19 & 2024-12-27 & 2025-01-03 & 6029.25 & -2205.0 & 35 & 2025-03-10 & 2025-03-17 & 5602.00 & 6470.0 \\
4 & 2023-10-19 & 2023-10-26 & 4295.00 & -5755.0 & 20 & 2024-12-30 & 2025-01-06 & 5953.25 & 3620.0 & 36 & 2025-03-11 & 2025-03-18 & 5593.50 & 3945.0 \\
5 & 2023-10-20 & 2023-10-27 & 4244.50 & -5405.0 & 21 & 2024-12-31 & 2025-01-07 & 5938.25 & 995.0 & 37 & 2025-03-12 & 2025-03-19 & 5620.50 & 6357.5 \\
6 & 2023-10-23 & 2023-10-30 & 4254.00 & -3605.0 & 22 & 2025-01-02 & 2025-01-09 & 5914.75 & -2667.5 & 38 & 2025-03-13 & 2025-03-20 & 5603.75 & 5545.0 \\
7 & 2023-10-24 & 2023-10-31 & 4268.25 & -3292.5 & 23 & 2025-01-03 & 2025-01-10 & 5985.25 & -6192.5 & 39 & 2025-03-14 & 2025-03-21 & 5681.00 & 1945.0 \\
8 & 2023-10-25 & 2023-11-01 & 4191.50 & 3732.5 & 24 & 2025-01-06 & 2025-01-13 & 6025.75 & -6667.5 & 40 & 2025-03-17 & 2025-03-24 & 5731.50 & 3995.0 \\
9 & 2023-10-26 & 2023-11-02 & 4180.00 & 7382.5 & 25 & 2025-01-07 & 2025-01-14 & 5958.25 & -3592.5 & 41 & 2025-03-18 & 2025-03-25 & 5672.50 & 8007.5 \\
10 & 2023-10-27 & 2023-11-03 & 4136.50 & 12095.0 & 26 & 2025-01-08 & 2025-01-15 & 5949.25 & 2245.0 & 42 & 2025-03-19 & 2025-03-26 & 5747.75 & -30.0 \\
11 & 2023-10-30 & 2023-11-06 & 4182.00 & 9720.0 & 27 & 2025-01-10 & 2025-01-17 & 5861.50 & 8545.0 & 43 & 2025-04-16 & 2025-04-23 & 5307.75 & 5207.5 \\
12 & 2024-04-23 & 2024-04-30 & 5116.25 & -2930.0 & 28 & 2025-01-13 & 2025-01-20 & 5892.50 & 10195.0 & 44 & 2025-04-17 & 2025-04-24 & 5326.50 & 11295.0 \\
13 & 2024-04-24 & 2024-05-01 & 5075.00 & -417.5 & 29 & 2025-01-14 & 2025-01-21 & 5886.50 & 10495.0 & 45 & 2025-04-21 & 2025-04-28 & 5205.25 & 17495.0 \\
14 & 2024-07-31 & 2024-08-07 & 5587.50 & -19267.5 & 30 & 2025-03-03 & 2025-03-10 & 5870.25 & -13417.5 & 46 & 2025-04-22 & 2025-04-29 & 5412.50 & 7945.0 \\
15 & 2024-08-01 & 2024-08-08 & 5461.00 & -5167.5 & 31 & 2025-03-04 & 2025-03-11 & 5834.00 & -12030.0 & 47 & 2025-04-23 & 2025-04-30 & 5412.00 & 11107.5 \\
\bottomrule
\end{tabular}%
}
\justify
\small{\textit{Notes:} This table provides the complete trade log for the Fixed-Size backtesting strategy, evaluated on the hold-out test set (July 2023 - June 2025). The strategy trades the E-mini S\&P 500 futures contract (ES). A long position of a single contract is initiated at the close on any day $t$ where the model's calibrated trough probability exceeds 5\%. Each position is held for a fixed period of 5 trading days. The P\&L is reported in US dollars and accounts for a round-trip transaction cost of \$5.00 per contract. The strategy rules are detailed in Section \ref{sec:econ_sig}.}
\end{table}

\begin{table}[htbp]
\centering
\caption{Full Trade Log: Pyramiding Strategy}
\label{tab:backtest_log_pyramid}
\resizebox{\textwidth}{!}{%
\begin{tabular}{rrrrrr|rrrrrr}
\toprule
\textbf{\#} & \textbf{Entry Date} & \textbf{Exit Date} & \textbf{Entry Price} & \textbf{Size} & \textbf{P\&L (\$)} & \textbf{\#} & \textbf{Entry Date} & \textbf{Exit Date} & \textbf{Entry Price} & \textbf{Size} & \textbf{P\&L (\$)} \\
\midrule
0 & 2023-10-13 & 2023-10-20 & 4351.50 & 1 & -5355.0 & 24 & 2025-01-06 & 2025-01-13 & 6025.75 & 6 & -40005.0 \\
1 & 2023-10-16 & 2023-10-23 & 4402.25 & 2 & -14835.0 & 25 & 2025-01-07 & 2025-01-14 & 5958.25 & 7 & -25147.5 \\
2 & 2023-10-17 & 2023-10-24 & 4392.00 & 3 & -18577.5 & 26 & 2025-01-08 & 2025-01-15 & 5949.25 & 8 & 17960.0 \\
3 & 2023-10-18 & 2023-10-25 & 4344.00 & 4 & -30520.0 & 27 & 2025-01-10 & 2025-01-17 & 5861.50 & 9 & 76905.0 \\
4 & 2023-10-19 & 2023-10-26 & 4295.00 & 5 & -28775.0 & 28 & 2025-01-13 & 2025-01-20 & 5892.50 & 10 & 101950.0 \\
5 & 2023-10-20 & 2023-10-27 & 4244.50 & 6 & -32430.0 & 29 & 2025-01-14 & 2025-01-21 & 5886.50 & 11 & 115445.0 \\
6 & 2023-10-23 & 2023-10-30 & 4254.00 & 7 & -25235.0 & 30 & 2025-03-03 & 2025-03-10 & 5870.25 & 1 & -13417.5 \\
7 & 2023-10-24 & 2023-10-31 & 4268.25 & 8 & -26340.0 & 31 & 2025-03-04 & 2025-03-11 & 5834.00 & 2 & -24060.0 \\
8 & 2023-10-25 & 2023-11-01 & 4191.50 & 9 & 33592.5 & 32 & 2025-03-05 & 2025-03-12 & 5845.25 & 3 & -33727.5 \\
9 & 2023-10-26 & 2023-11-02 & 4180.00 & 10 & 73825.0 & 33 & 2025-03-06 & 2025-03-13 & 5764.75 & 4 & -32220.0 \\
10 & 2023-10-27 & 2023-11-03 & 4136.50 & 11 & 133045.0 & 34 & 2025-03-07 & 2025-03-14 & 5771.75 & 5 & -22712.5 \\
11 & 2023-10-30 & 2023-11-06 & 4182.00 & 12 & 116640.0 & 35 & 2025-03-10 & 2025-03-17 & 5602.00 & 6 & 38820.0 \\
12 & 2024-04-23 & 2024-04-30 & 5116.25 & 1 & -2930.0 & 36 & 2025-03-11 & 2025-03-18 & 5593.50 & 7 & 27615.0 \\
13 & 2024-04-24 & 2024-05-01 & 5075.00 & 2 & -835.0 & 37 & 2025-03-12 & 2025-03-19 & 5620.50 & 8 & 50860.0 \\
14 & 2024-07-31 & 2024-08-07 & 5587.50 & 1 & -19267.5 & 38 & 2025-03-13 & 2025-03-20 & 5603.75 & 9 & 49905.0 \\
15 & 2024-08-01 & 2024-08-08 & 5461.00 & 2 & -10335.0 & 39 & 2025-03-14 & 2025-03-21 & 5681.00 & 10 & 19450.0 \\
16 & 2024-08-02 & 2024-08-09 & 5358.00 & 3 & 1035.0 & 40 & 2025-03-17 & 2025-03-24 & 5731.50 & 11 & 43945.0 \\
17 & 2024-08-05 & 2024-08-12 & 5273.75 & 4 & 19730.0 & 41 & 2025-03-18 & 2025-03-25 & 5672.50 & 12 & 96090.0 \\
18 & 2024-08-06 & 2024-08-13 & 5243.50 & 5 & 53537.5 & 42 & 2025-03-19 & 2025-03-26 & 5747.75 & 13 & -390.0 \\
19 & 2024-12-27 & 2025-01-03 & 6029.25 & 1 & -2205.0 & 43 & 2025-04-16 & 2025-04-23 & 5307.75 & 1 & 5207.5 \\
20 & 2024-12-30 & 2025-01-06 & 5953.25 & 2 & 7240.0 & 44 & 2025-04-17 & 2025-04-24 & 5326.50 & 2 & 22590.0 \\
21 & 2024-12-31 & 2025-01-07 & 5938.25 & 3 & 2985.0 & 45 & 2025-04-21 & 2025-04-28 & 5205.25 & 3 & 52485.0 \\
22 & 2025-01-02 & 2025-01-09 & 5914.75 & 4 & -10670.0 & 46 & 2025-04-22 & 2025-04-29 & 5412.50 & 4 & 31780.0 \\
23 & 2025-01-03 & 2025-01-10 & 5985.25 & 5 & -30962.5 & 47 & 2025-04-23 & 2025-04-30 & 5412.00 & 5 & 55537.5 \\
\bottomrule
\end{tabular}%
}
\justify
\small{\textit{Notes:} This table provides the complete trade log for the Pyramiding backtesting strategy, evaluated on the hold-out test set (July 2023 - June 2025). The strategy trades the E-mini S\&P 500 futures contract (ES). A long position is initiated at the close on any day $t$ where the model's calibrated trough probability exceeds 5\%. The position size is increased with each consecutive day a signal is active; a trade of size $N$ is placed on the $N^{th}$ consecutive signal day, as shown in the 'Size' column. Each position is held for a fixed period of 5 trading days. The P\&L is reported in US dollars and accounts for a round-trip transaction cost of \$5.00 per contract. The strategy rules are detailed in Section \ref{sec:econ_sig}.}
\end{table}

\subsection{Complete Holding Period Sensitivity Analysis Results}
\label{sec:appendix_sensitivity_full}

This appendix provides the full, unabridged results of the holding period sensitivity analysis referenced in Section \ref{sec:econ_sig}. Table \ref{tab:backtest_sensitivity_full} details the performance metrics for all tested holding periods from 5 to 20 days.

\begin{table}[!htbp]
\centering
\caption{Complete Backtest Sensitivity Analysis Summary}
\label{tab:backtest_sensitivity_full}
\resizebox{\textwidth}{!}{%
\begin{tabular}{llrrrrr}
\toprule
\textbf{Holding Period} & \textbf{Strategy} & \textbf{Total Net P\&L} & \textbf{Sharpe Ratio (Ann.)} & \textbf{Profit Factor} & \textbf{Max Drawdown} & \textbf{Max Drawdown (\%)} \\
\midrule
\multirow{2}{*}{5 Days} & Fixed-Size & \$31,247.50 & 0.38 & 1.22 & (\$52,682.50) & 55.66\% \\
& Pyramiding & \$797,222.50 & 1.62 & 2.77 & (\$176,712.50) & 186.71\% \\
\midrule
\multirow{2}{*}{7 Days} & Fixed-Size & \$112,385.00 & 1.23 & 1.93 & (\$39,287.50) & 41.00\% \\
& Pyramiding & \$1,180,760.00 & 2.00 & 3.95 & (\$135,325.00) & 141.21\% \\
\midrule
\multirow{2}{*}{10 Days} & Fixed-Size & \$200,985.00 & 2.01 & 3.00 & (\$25,000.00) & 10.76\% \\
& Pyramiding & \$1,404,622.50 & 2.18 & 4.42 & (\$239,230.00) & 15.74\% \\
\midrule
\multirow{2}{*}{12 Days} & Fixed-Size & \$235,210.00 & 2.03 & 3.34 & (\$56,052.50) & 18.37\% \\
& Pyramiding & \$1,165,810.00 & 1.21 & 2.50 & (\$694,205.00) & 40.40\% \\
\midrule
\multirow{2}{*}{15 Days} & Fixed-Size & \$249,235.00 & 1.68 & 2.79 & (\$132,027.50) & 36.16\% \\
& Pyramiding & \$772,910.00 & 0.63 & 1.56 & (\$1,343,805.00) & 72.77\% \\
\midrule
\multirow{2}{*}{17 Days} & Fixed-Size & \$236,472.50 & 1.48 & 2.29 & (\$180,727.50) & 47.02\% \\
& Pyramiding & \$728,035.00 & 0.60 & 1.47 & (\$1,551,217.50) & 78.94\% \\
\midrule
\multirow{2}{*}{20 Days} & Fixed-Size & \$217,385.00 & 1.23 & 1.95 & (\$229,240.00) & 57.59\% \\
& Pyramiding & \$735,522.50 & 0.63 & 1.45 & (\$1,634,867.50) & 80.06\% \\
\bottomrule
\end{tabular}
}
\justify
\small{\textit{Notes:} This table presents the complete results of the holding period sensitivity analysis conducted on the hold-out test set. Performance metrics are shown for all tested holding periods. A curated version of this table highlighting the key findings is presented in the main text in Table \ref{tab:backtest_sensitivity_summary}.}
\end{table}

\section{Full DML Estimation and Sensitivity Analysis Results}
\label{sec:appendix_results}

This appendix contains the complete set of treatment variables for which the Double/Debiased Machine Learning analysis yielded a statistically significant causal estimate ($p < 0.05$) that is also robust to the formal sensitivity analysis of \citet{Cinelli2020}. Table \ref{tab:dml_plr_full_results} lists the robust findings from the baseline DML-PLR model. Table \ref{tab:dml_ape_full_results} lists the robust findings from our primary DML-APE model.

\begin{table}[!t]
\centering
\caption{Complete Robust Causal Estimates from the DML-PLR Model}
\label{tab:dml_plr_full_results}
\resizebox{\textwidth}{!}{%
\begin{tabular}{lrrrrrrr}
\toprule
\textbf{Treatment Variable} & \textbf{Coeff. ($\hat{\theta}$)} & \textbf{$p$-value} & \textbf{bias\_phi} & \textbf{Adj. 95\% CI Lower} & \textbf{Adj. 95\% CI Upper} & \textbf{Benchmark $R^2_Y$} & \textbf{Benchmark $R^2_D$} \\
\midrule
ffr\_slope\_trend\_z\_scaled\_last & -0.0157 & 0.0000 & 0.0000 & -0.0216 & -0.0098 & 0.0520 & 0.0000 \\
ffr\_slope\_roc63\_scaled\_last & -0.0228 & 0.0000 & 0.0048 & -0.0368 & -0.0087 & 0.0518 & 0.0159 \\
credit\_spread\_scaled\_std & -0.0524 & 0.0000 & 0.0000 & -0.0743 & -0.0306 & 0.0433 & 0.0000 \\
credit\_spread\_trend\_z\_scaled\_std & -0.0522 & 0.0000 & 0.0000 & -0.0747 & -0.0297 & 0.0424 & 0.0000 \\
fx\_rv\_6j\_21d\_trend\_z\_scaled\_mean & -0.0143 & 0.0000 & 0.0000 & -0.0207 & -0.0078 & 0.0664 & 0.0000 \\
ffr\_slope\_roc63\_scaled\_mean & -0.0122 & 0.0001 & 0.0000 & -0.0181 & -0.0062 & 0.0520 & 0.0000 \\
ffr\_slope\_scaled\_last & -0.0154 & 0.0001 & 0.0000 & -0.0230 & -0.0079 & 0.0518 & 0.0000 \\
amihud\_illiquidity\_trend\_z\_scaled\_std & -0.0608 & 0.0001 & 0.0000 & -0.0917 & -0.0299 & 0.0569 & 0.0000 \\
ffr\_slope\_scaled\_mean & -0.0145 & 0.0003 & 0.0000 & -0.0224 & -0.0066 & 0.0520 & 0.0000 \\
fx\_rv\_6j\_21d\_trend\_z\_scaled\_last & -0.0153 & 0.0005 & 0.0000 & -0.0238 & -0.0067 & 0.0656 & 0.0000 \\
ffr\_slope\_scaled\_trend & -0.1436 & 0.0010 & 0.0000 & -0.2293 & -0.0579 & 0.0515 & 0.0000 \\
fx\_rv\_6j\_21d\_scaled\_mean & -0.0128 & 0.0025 & 0.0038 & -0.0248 & -0.0007 & 0.0656 & 0.0069 \\
pcr\_oi\_roc63\_scaled\_std & -0.0549 & 0.0057 & 0.0000 & -0.0938 & -0.0160 & 0.0932 & 0.0000 \\
risk\_neutral\_skewness\_scaled\_trend & -0.1346 & 0.0069 & 0.0000 & -0.2321 & -0.0370 & 0.0448 & 0.0000 \\
risk\_neutral\_skewness\_trend\_z\_scaled\_trend & -0.1533 & 0.0082 & 0.0000 & -0.2669 & -0.0397 & 0.0446 & 0.0000 \\
flow\_concentration\_10d\_scaled\_std & 0.0636 & 0.0124 & 0.0000 & 0.0138 & 0.1134 & 0.0533 & 0.0000 \\
fx\_rv\_6j\_21d\_scaled\_last & -0.0096 & 0.0135 & 0.0000 & -0.0172 & -0.0020 & 0.0682 & 0.0000 \\
ffr\_basis\_roc63\_scaled\_mean & 0.0111 & 0.0177 & 0.0000 & 0.0019 & 0.0202 & 0.0520 & 0.0000 \\
risk\_neutral\_kurtosis\_trend\_z\_scaled\_mean & 0.0186 & 0.0196 & 0.0000 & 0.0030 & 0.0342 & 0.0448 & 0.0000 \\
flow\_concentration\_10d\_trend\_z\_scaled\_mean & -0.0093 & 0.0202 & 0.0000 & -0.0172 & -0.0015 & 0.0523 & 0.0000 \\
flow\_concentration\_10d\_trend\_z\_scaled\_std & 0.0404 & 0.0219 & 0.0000 & 0.0059 & 0.0750 & 0.0546 & 0.0000 \\
flow\_concentration\_10d\_roc63\_scaled\_std & 0.0515 & 0.0228 & 0.0000 & 0.0072 & 0.0958 & 0.0594 & 0.0000 \\
ffr\_basis\_roc63\_scaled\_last & 0.0081 & 0.0260 & 0.0000 & 0.0010 & 0.0152 & 0.0526 & 0.0000 \\
ffr\_slope\_trend\_z\_scaled\_trend & -0.0834 & 0.0306 & 0.0000 & -0.1590 & -0.0078 & 0.0520 & 0.0000 \\
risk\_neutral\_kurtosis\_scaled\_mean & 0.0173 & 0.0306 & 0.0000 & 0.0016 & 0.0329 & 0.0448 & 0.0000 \\
risk\_neutral\_skewness\_roc63\_scaled\_trend & -0.0946 & 0.0359 & 0.0000 & -0.1829 & -0.0062 & 0.0465 & 0.0000 \\
pcr\_oi\_trend\_z\_scaled\_std & -0.0337 & 0.0420 & 0.0000 & -0.0662 & -0.0012 & 0.1126 & 0.0000 \\
\bottomrule
\end{tabular}%
}
\justify
\small{\textit{Notes:} This table reports the complete set of statistically significant ($p < 0.05$) causal estimates from the Double/Debiased Machine Learning Partially Linear Regression (DML-PLR) model that are robust to unobserved confounding. The analysis uses daily data from April 2013 to June 2025 ($N=3068$). The dependent variable is a binary indicator for a market trough. For each treatment variable listed, the model includes all other engineered features as high-dimensional controls, subject to the exclusion protocol described in Section 7.3.1. Coefficients ($\hat{\theta}$) represent the estimated linear causal effect of a one-unit change in the treatment variable on the probability of a market trough. Adjusted 95\% confidence intervals and p-values are based on robust standard errors from the DML procedure. The `bias\_phi` column reports a measure of the estimation bias from the nuisance functions.

Robustness is assessed using the formal sensitivity analysis of \citet{Cinelli2020}. Benchmark $R^2_Y$ and Benchmark $R^2_D$ report the out-of-sample partial $R^2$ of the outcome and the treatment explained by the observed confounders, respectively. These values serve as a benchmark for the plausible strength of an unobserved confounder. The results are deemed robust if the adjusted 95\% confidence interval, which accounts for potential bias from a hypothetical confounder as strong as the observed ones, still excludes zero.}

\end{table}

\clearpage

\begin{table}[t]
\centering
\caption{Complete Robust Causal Estimates from the DML-APE Model}
\label{tab:dml_ape_full_results}
\resizebox{\textwidth}{!}{%
\begin{tabular}{lrrrrrrr}
\toprule
\textbf{Treatment Variable} & \textbf{Coeff. ($\hat{\theta}$)} & \textbf{$p$-value} & \textbf{bias\_phi} & \textbf{Adj. 95\% CI Lower} & \textbf{Adj. 95\% CI Upper} & \textbf{Benchmark $R^2_Y$} & \textbf{Benchmark $R^2_D$} \\
\midrule
\seqsplit{fx\_rv\_6j\_21d\_roc63\_scaled\_std} & 0.0057 & 0.0000 & 0.0000 & 0.0046 & 0.0069 & 0.1085 & 0.0000 \\
\seqsplit{risk\_neutral\_kurtosis\_trend\_z\_scaled\_std} & 0.0485 & 0.0000 & 0.0000 & 0.0413 & 0.0557 & 0.0366 & 0.0000 \\
\seqsplit{fx\_rv\_6j\_21d\_trend\_z\_scaled\_std} & 0.0072 & 0.0000 & 0.0000 & 0.0061 & 0.0083 & 0.1085 & 0.0000 \\
\seqsplit{risk\_neutral\_kurtosis\_roc63\_scaled\_std} & 0.0382 & 0.0000 & 0.0000 & 0.0339 & 0.0425 & 0.0366 & 0.0000 \\
\seqsplit{ffr\_slope\_trend\_z\_scaled\_std} & 0.0092 & 0.0000 & 0.0000 & 0.0083 & 0.0100 & 0.0427 & 0.0000 \\
\seqsplit{fx\_rv\_6e\_21d\_scaled\_std} & 0.0100 & 0.0000 & 0.0000 & 0.0089 & 0.0110 & 0.0913 & 0.0000 \\
\seqsplit{amihud\_illiquidity\_trend\_z\_scaled\_std} & 0.0160 & 0.0000 & 0.0000 & 0.0129 & 0.0191 & 0.0375 & 0.0000 \\
\seqsplit{risk\_neutral\_skewness\_trend\_z\_scaled\_std} & 0.0343 & 0.0000 & 0.0000 & 0.0301 & 0.0384 & 0.0366 & 0.0000 \\
\seqsplit{risk\_neutral\_skewness\_scaled\_std} & 0.0359 & 0.0000 & 0.0000 & 0.0320 & 0.0397 & 0.0366 & 0.0000 \\
\seqsplit{risk\_neutral\_skewness\_roc63\_scaled\_std} & 0.0399 & 0.0000 & 0.0000 & 0.0320 & 0.0479 & 0.0366 & 0.0000 \\
\seqsplit{pcr\_oi\_trend\_z\_scaled\_std} & 0.0213 & 0.0000 & 0.0000 & 0.0196 & 0.0231 & 0.1064 & 0.0000 \\
\seqsplit{fx\_rv\_6e\_21d\_roc63\_scaled\_mean} & -0.0016 & 0.0000 & 0.0000 & -0.0020 & -0.0012 & 0.0913 & 0.0000 \\
\seqsplit{ffr\_slope\_roc63\_scaled\_std} & 0.0063 & 0.0000 & 0.0000 & 0.0057 & 0.0070 & 0.0427 & 0.0000 \\
\seqsplit{pcr\_oi\_scaled\_std} & 0.0191 & 0.0000 & 0.0000 & 0.0178 & 0.0203 & 0.1064 & 0.0000 \\
\seqsplit{fx\_rv\_6e\_21d\_roc63\_scaled\_std} & 0.0061 & 0.0000 & 0.0000 & 0.0048 & 0.0074 & 0.0913 & 0.0000 \\
\seqsplit{pcr\_oi\_roc63\_scaled\_std} & 0.0241 & 0.0000 & 0.0000 & 0.0214 & 0.0268 & 0.1064 & 0.0000 \\
\seqsplit{ffr\_slope\_scaled\_std} & 0.0103 & 0.0000 & 0.0000 & 0.0099 & 0.0107 & 0.0427 & 0.0000 \\
\seqsplit{ffr\_slope\_scaled\_last} & 0.0177 & 0.0000 & 0.0000 & 0.0146 & 0.0209 & 0.0427 & 0.0000 \\
\seqsplit{fx\_rv\_6j\_21d\_scaled\_std} & 0.0051 & 0.0000 & 0.0000 & 0.0039 & 0.0063 & 0.1085 & 0.0000 \\
\seqsplit{risk\_neutral\_kurtosis\_scaled\_std} & 0.0957 & 0.0000 & 0.0000 & 0.0726 & 0.1189 & 0.0366 & 0.0000 \\
\seqsplit{fx\_momentum\_6e\_21d\_trend\_z\_scaled\_std} & 0.0047 & 0.0000 & 0.0000 & 0.0035 & 0.0058 & 0.0913 & 0.0000 \\
\seqsplit{ffr\_slope\_scaled\_trend} & -0.0073 & 0.0000 & 0.0000 & -0.0092 & -0.0054 & 0.0427 & 0.0000 \\
\seqsplit{gex\_oi\_trend\_z\_scaled\_std} & 0.0773 & 0.0000 & 0.0186 & 0.0381 & 0.1165 & 0.0331 & 0.0391 \\
\seqsplit{fx\_momentum\_6e\_21d\_scaled\_std} & 0.0038 & 0.0000 & 0.0000 & 0.0028 & 0.0049 & 0.0913 & 0.0000 \\
\seqsplit{fx\_rv\_6j\_21d\_trend\_z\_scaled\_last} & -0.0018 & 0.0000 & 0.0000 & -0.0023 & -0.0012 & 0.1085 & 0.0000 \\
\seqsplit{fx\_momentum\_6j\_21d\_trend\_z\_scaled\_std} & 0.0026 & 0.0000 & 0.0000 & 0.0018 & 0.0035 & 0.1085 & 0.0000 \\
\seqsplit{fx\_rv\_6e\_21d\_roc63\_scaled\_last} & -0.0012 & 0.0000 & 0.0000 & -0.0017 & -0.0008 & 0.0913 & 0.0000 \\
\seqsplit{ffr\_slope\_trend\_z\_scaled\_trend} & -0.0066 & 0.0000 & 0.0000 & -0.0090 & -0.0042 & 0.0427 & 0.0000 \\
\seqsplit{fx\_rv\_6e\_21d\_trend\_z\_scaled\_mean} & -0.0012 & 0.0000 & 0.0000 & -0.0017 & -0.0008 & 0.0913 & 0.0000 \\
\seqsplit{fx\_rv\_6e\_21d\_trend\_z\_scaled\_last} & -0.0009 & 0.0000 & 0.0000 & -0.0013 & -0.0006 & 0.0913 & 0.0000 \\
\seqsplit{fx\_rv\_6j\_21d\_trend\_z\_scaled\_trend} & 0.0037 & 0.0000 & 0.0000 & 0.0020 & 0.0054 & 0.1085 & 0.0000 \\
\seqsplit{fx\_rv\_6j\_21d\_trend\_z\_scaled\_mean} & -0.0017 & 0.0001 & 0.0000 & -0.0025 & -0.0009 & 0.1085 & 0.0000 \\
\seqsplit{ffr\_slope\_roc63\_scaled\_last} & -0.0006 & 0.0001 & 0.0000 & -0.0009 & -0.0003 & 0.0427 & 0.0000 \\
\seqsplit{fx\_rv\_6e\_21d\_scaled\_mean} & -0.0009 & 0.0003 & 0.0000 & -0.0014 & -0.0004 & 0.0913 & 0.0000 \\
\seqsplit{flow\_concentration\_10d\_trend\_z\_scaled\_mean} & -0.0013 & 0.0003 & 0.0000 & -0.0021 & -0.0006 & 0.0424 & 0.0000 \\
\seqsplit{risk\_neutral\_kurtosis\_scaled\_mean} & 0.0051 & 0.0008 & 0.0000 & 0.0021 & 0.0081 & 0.0366 & 0.0000 \\
\seqsplit{fx\_rv\_6e\_21d\_trend\_z\_scaled\_trend} & 0.0024 & 0.0009 & 0.0000 & 0.0010 & 0.0038 & 0.0913 & 0.0000 \\
\seqsplit{flow\_concentration\_10d\_scaled\_std} & 0.0021 & 0.0013 & 0.0000 & 0.0008 & 0.0034 & 0.0424 & 0.0000 \\
\seqsplit{ffr\_basis\_roc63\_scaled\_trend} & -0.0020 & 0.0034 & 0.0000 & -0.0033 & -0.0006 & 0.0427 & 0.0000 \\
\seqsplit{ffr\_slope\_scaled\_mean} & 0.0011 & 0.0085 & 0.0000 & 0.0003 & 0.0019 & 0.0427 & 0.0000 \\
\seqsplit{ffr\_basis\_scaled\_last} & 0.0007 & 0.0096 & 0.0000 & 0.0002 & 0.0012 & 0.0427 & 0.0000 \\
\seqsplit{vrp\_roc63\_scaled\_std} & -0.0021 & 0.0099 & 0.0000 & -0.0036 & -0.0005 & 0.0367 & 0.0000 \\
\seqsplit{flow\_concentration\_10d\_roc63\_scaled\_std} & 0.0018 & 0.0114 & 0.0000 & 0.0004 & 0.0031 & 0.0424 & 0.0000 \\
\seqsplit{risk\_neutral\_kurtosis\_trend\_z\_scaled\_trend} & -0.0060 & 0.0233 & 0.0000 & -0.0111 & -0.0008 & 0.0366 & 0.0000 \\
\seqsplit{ffr\_basis\_roc63\_scaled\_last} & -0.0003 & 0.0299 & 0.0000 & -0.0006 & 0.0000 & 0.0427 & 0.0000 \\
\seqsplit{risk\_neutral\_kurtosis\_scaled\_trend} & -0.0043 & 0.0313 & 0.0000 & -0.0082 & -0.0004 & 0.0366 & 0.0000 \\
\seqsplit{ffr\_basis\_scaled\_mean} & 0.0006 & 0.0345 & 0.0000 & 0.0000 & 0.0011 & 0.0427 & 0.0000 \\
\seqsplit{flow\_concentration\_10d\_trend\_z\_scaled\_std} & 0.0011 & 0.0360 & 0.0000 & 0.0001 & 0.0021 & 0.0424 & 0.0000 \\
\bottomrule
\end{tabular}%
}
\justify
\small{\textit{Notes:} This table reports the complete set of statistically significant ($p < 0.05$) causal estimates from the Double/Debiased Machine Learning Average Partial Effect (DML-APE) model that are robust to unobserved confounding. The analysis uses daily data from April 2013 to June 2025 ($N=3068$). The dependent variable is a binary indicator for a market trough. For each treatment variable, all other engineered features are included as high-dimensional controls, subject to the exclusion protocol in Section 7.3.1.

The coefficient ($\hat{\theta}$) is the Average Partial Effect (APE), representing the average change in trough probability for a one-unit change in the treatment, averaged over the data distribution. Following the procedure in Section 7.2, the point estimate is the median of the Neyman-orthogonal scores, and the 95\% confidence intervals and p-values are derived from a non-parametric bootstrap of these scores. This method is chosen for its robustness to outliers in the score function, as justified in \ref{sec:appendix_ape_median}.

All findings are validated using the sensitivity analysis framework of \citet{Cinelli2020}. Benchmark $R^2_Y$ and Benchmark $R^2_D$ report the out-of-sample partial $R^2$ of the outcome and the treatment explained by the observed confounders, respectively. These values serve as a benchmark for the plausible strength of an unobserved confounder. The results are deemed robust if the adjusted 95\% confidence interval, which accounts for potential bias from a hypothetical confounder as strong as the observed ones, still excludes zero.}

\end{table}

\section{Deriving the Neyman Orthogonal Score for the APE}

\label{sec:appendixapederiv}

In this appendix we provide a formal derivation of the Neyman-orthogonal score function for the Average Partial Effect (APE) in a general non-parametric interactive model. We state the general form of the score, then provide a rigorous proof of its Neyman-orthogonality property with respect to both the outcome and treatment models, clarifying all necessary assumptions. Finally, we derive the practical and implementable score function that results from a semi-parametric Gaussian assumption for the treatment model, providing a complete theoretical basis for its use in debiased machine learning.

\subsection{Model Setup, Assumptions, and Parameter of Interest}

We consider a structural model where an outcome variable $\mathbf{Y}$ is determined by a continuous scalar treatment $\mathbf{D}$ and a vector of confounding variables $\mathbf{X}$ through a general, non-separable function $g_0$:
$$
\mathbf{Y} = g_0(\mathbf{D}, \mathbf{X}) + \mathbf{U}
$$
The core causal assumption is unconfoundedness, which posits that the treatment assignment is independent of the potential outcomes, conditional on the covariates $\mathbf{X}$. This implies that $\mathbf{U}$ is mean-independent of $\mathbf{D}$ conditional on $\mathbf{X}$, which gives the key moment condition $\mathbb{E}[\mathbf{U} | \mathbf{D}, \mathbf{X}] = 0$. This justifies defining the conditional expectation function as $l_0(\mathbf{D}, \mathbf{X}) = \mathbb{E}[\mathbf{Y} | \mathbf{D}, \mathbf{X}]$. Let $p_0(\mathbf{D}|\mathbf{X})$ denote the true conditional probability density function (PDF) of the treatment $\mathbf{D}$ given the confounders $\mathbf{X}$.

For the derivation that follows, we make the following assumptions:
\begin{itemize}
    \item \textbf{Regularity Conditions}:
    \begin{itemize}
        \item The conditional expectation function $l_0(d, \mathbf{X})$ is differentiable with respect to its first argument $d$ for all $(\mathbf{D}, \mathbf{X})$ in the support of the data distribution.
        \item All expectations presented in this document exist and are finite.
        \item For any valid perturbation function $r_l(d, \mathbf{X})$ in the Gateaux derivative path, the boundary condition $\lim_{d \to \pm\infty} r_l(d, \mathbf{X}) p_0(d|\mathbf{X}) = 0$ holds. This ensures the validity of the integration by parts used in the proof of orthogonality.
        \item The perturbed density paths are sufficiently smooth to allow for the interchange of partial derivatives (Clairaut's Theorem).
    \end{itemize}
\end{itemize}
The causal parameter of interest is the \textbf{Average Partial Effect (APE)}, denoted $\theta_0$, which is the expected partial derivative of the outcome model with respect to the treatment:
$$
\theta_0 = \mathbb{E}\left[ \frac{\partial l_0(\mathbf{D}, \mathbf{X})}{\partial \mathbf{D}} \right]
$$
The expectation is taken over the joint distribution of $(\mathbf{D}, \mathbf{X})$.

\subsection{The General Neyman-Orthogonal Score}

To achieve a $\sqrt{n}$-consistent and asymptotically normal estimator for $\theta_0$, we use a debiased estimation approach centered on a Neyman-orthogonal score function. This score corrects for the bias introduced by regularization in high-dimensional or non-parametric estimation of the nuisance functions.

\begin{prop}
The Neyman-orthogonal score function for the APE parameter $\theta_0$ is given by:
$$
\psi(\mathbf{W}; \theta, \eta) = \frac{\partial l(\mathbf{D}, \mathbf{X})}{\partial \mathbf{D}} - \theta - \frac{1}{p(\mathbf{D}|\mathbf{X})}\frac{\partial p(\mathbf{D}|\mathbf{X})}{\partial \mathbf{D}} \left( \mathbf{Y} - l(\mathbf{D}, \mathbf{X}) \right)
$$
where $\mathbf{W} = (\mathbf{Y}, \mathbf{D}, \mathbf{X})$ is the observable data and $\eta = (l, p)$ is the set of nuisance functions.
\end{prop}

\begin{proof}[Proof of Neyman-Orthogonality]
A score function is Neyman-orthogonal if its Gateaux derivative with respect to the nuisance functions, evaluated at the true parameters, is zero. Let $\eta_0 = (l_0, p_0)$ and let $\theta_0$ be the true parameter values. We must show that for valid directional perturbations $[r_l]$ and $[r_p]$, the following conditions hold:
$$
\left. \frac{\partial}{\partial t} \mathbb{E}\left[\psi(\mathbf{W}; \theta_0, l_0+t[r_l], p_0)\right] \right|_{t=0} = 0
$$
$$
\left. \frac{\partial}{\partial t} \mathbb{E}\left[\psi(\mathbf{W}; \theta_0, l_0, p_0+t[r_p])\right] \right|_{t=0} = 0
$$

\begin{itemize}
    \item \textbf{1. Orthogonality with respect to the outcome model $l$}: Differentiating the expected score with respect to the path parameter $t$ and evaluating at $t=0$ yields:
    $$
    \mathbb{E}\left[ \frac{\partial r_l(\mathbf{D}, \mathbf{X})}{\partial \mathbf{D}} - \frac{1}{p_0(\mathbf{D}|\mathbf{X})}\frac{\partial p_0(\mathbf{D}|\mathbf{X})}{\partial \mathbf{D}} \left( -r_l(\mathbf{D}, \mathbf{X}) \right) \right] = \mathbb{E}\left[ \frac{\partial r_l}{\partial \mathbf{D}} \right] + \mathbb{E}\left[ \frac{p'_{0,\mathbf{D}}}{p_0} r_l \right]
    $$
    We analyze the second term by conditioning on $\mathbf{X}$ and using integration by parts over the domain of $\mathbf{D}$. Let $d$ be a realization of $\mathbf{D}$.
    $$
    \mathbb{E}_{\mathbf{D}|\mathbf{X}}\left[ \frac{p'_{0,\mathbf{D}}(d|\mathbf{X})}{p_0(d|\mathbf{X})} r_l(d, \mathbf{X}) \right] = \int \frac{p'_{0,\mathbf{D}}(d|\mathbf{X})}{p_0(d|\mathbf{X})} r_l(d, \mathbf{X}) p_0(d|\mathbf{X}) dd = \int p'_{0,\mathbf{D}}(d|\mathbf{X}) r_l(d, \mathbf{X}) dd
    $$
    Applying integration by parts, $\int u \, dv = uv - \int v \, du$, with $u = r_l(d, \mathbf{X})$ and $dv = p'_{0,\mathbf{D}}(d|\mathbf{X}) dd$, gives:
    $$
    \left[ r_l(d, \mathbf{X}) p_0(d|\mathbf{X}) \right]_{-\infty}^{\infty} - \int p_0(d|\mathbf{X}) \frac{\partial r_l(d, \mathbf{X})}{\partial d} dd
    $$
    Per our regularity conditions, the boundary term is zero. The expression simplifies to:
    $$
    - \int p_0(d|\mathbf{X}) \frac{\partial r_l}{\partial d} dd = - \mathbb{E}_{\mathbf{D}|\mathbf{X}}\left[\frac{\partial r_l}{\partial \mathbf{D}}\right]
    $$
    Substituting this back into the full expectation, the Gateaux derivative is:
    $$
    \mathbb{E}\left[ \frac{\partial r_l}{\partial \mathbf{D}} \right] - \mathbb{E}\left[ \mathbb{E}_{\mathbf{D}|\mathbf{X}}\left[\frac{\partial r_l}{\partial \mathbf{D}}\right] \right] = \mathbb{E}\left[ \frac{\partial r_l}{\partial \mathbf{D}} \right] - \mathbb{E}\left[ \frac{\partial r_l}{\partial \mathbf{D}} \right] = 0
    $$
    The score is therefore orthogonal to perturbations in $l$.

    \item \textbf{2. Orthogonality with respect to the density model $p$}: The Gateaux derivative must be computed by taking the derivative with respect to $t$ inside the expectation, as the expectation is over the true, fixed data distribution. Let $p_t = p_0 + t[r_p]$.
    $$
    \mathcal{D}_p = \left. \frac{\partial}{\partial t} \mathbb{E}_{\mathbf{W}}\left[\psi(\mathbf{W}; \theta_0, l_0, p_t)\right] \right|_{t=0}
    $$
    The only term in $\psi$ that depends on the path $p_t$ is the correction term.
    $$
    \mathcal{D}_p = \mathbb{E}_{\mathbf{W}}\left[ \left. \frac{\partial}{\partial t} \left( - \frac{\partial \log p_t(\mathbf{D}|\mathbf{X})}{\partial \mathbf{D}} \left( \mathbf{Y} - l_0(\mathbf{D}, \mathbf{X}) \right) \right) \right|_{t=0} \right]
    $$
    Assuming the perturbed path is sufficiently smooth, we can interchange the derivatives with respect to $t$ and $\mathbf{D}$ (Clairaut's Theorem):
    $$
    \frac{\partial}{\partial t} \left(\frac{\partial \log p_t}{\partial \mathbf{D}}\right) = \frac{\partial}{\partial \mathbf{D}} \left(\frac{\partial \log p_t}{\partial t}\right) = \frac{\partial}{\partial \mathbf{D}} \left(\frac{r_p}{p_t}\right)
    $$
    Evaluating this at $t=0$ (where $p_t=p_0$):
    $$
    \left. \frac{\partial}{\partial t} \left(\frac{\partial \log p_t}{\partial \mathbf{D}}\right) \right|_{t=0} = \frac{\partial}{\partial \mathbf{D}} \left(\frac{r_p(\mathbf{D}|\mathbf{X})}{p_0(\mathbf{D}|\mathbf{X})}\right)
    $$
    Substituting back into the expectation for the Gateaux derivative gives:
    $$
    \mathcal{D}_p = \mathbb{E}_{\mathbf{W}}\left[ - \frac{\partial}{\partial \mathbf{D}} \left(\frac{r_p(\mathbf{D}|\mathbf{X})}{p_0(\mathbf{D}|\mathbf{X})}\right) \left( \mathbf{Y} - l_0(\mathbf{D}, \mathbf{X}) \right) \right]
    $$
    We now apply the Law of Iterated Expectations, conditioning on $(\mathbf{D}, \mathbf{X})$:
    $$
    \mathcal{D}_p = \mathbb{E}_{(\mathbf{D}, \mathbf{X})}\left[ - \frac{\partial}{\partial \mathbf{D}} \left(\frac{r_p}{p_0}\right) \mathbb{E}[\mathbf{Y} - l_0(\mathbf{D}, \mathbf{X}) | \mathbf{D}, \mathbf{X}] \right]
    $$
    From the model setup, $\mathbb{E}[\mathbf{Y} - l_0(\mathbf{D}, \mathbf{X}) | \mathbf{D}, \mathbf{X}] = \mathbb{E}[\mathbf{U}|\mathbf{D}, \mathbf{X}] = 0$. Thus:
    $$
    \mathcal{D}_p = \mathbb{E}_{(\mathbf{D}, \mathbf{X})}\left[ - \frac{\partial}{\partial \mathbf{D}} \left(\frac{r_p}{p_0}\right) \cdot 0 \right] = 0
    $$
    The score is therefore orthogonal to perturbations in $p$.
\end{itemize}
\end{proof}

\subsection{A Practical Score via a Semi-Parametric Assumption}

The general score in Proposition 1 is difficult to implement, as it requires non-parametric estimation of a conditional density and its derivative. A practical score is obtainable by imposing a semi-parametric structure on the treatment model.

\begin{assumption}[Heteroskedastic Gaussian Treatment Model]
The conditional distribution of the treatment $\mathbf{D}$ given the confounders $\mathbf{X}$ is Gaussian, with mean function $m_0(\mathbf{X})$ and variance function $v_0(\mathbf{X})$:
$$
\mathbf{D} | \mathbf{X} \sim \mathcal{N}(m_0(\mathbf{X}), v_0(\mathbf{X}))
$$
\end{assumption}

\begin{prop}
Under Assumption 1, the general Neyman-orthogonal score from Proposition 1 simplifies to:
$$
\psi(\mathbf{W}; \theta, \eta) = \frac{\partial l(\mathbf{D}, \mathbf{X})}{\partial \mathbf{D}} - \theta + \frac{(\mathbf{D} - m(\mathbf{X}))}{v(\mathbf{X})} \left( \mathbf{Y} - l(\mathbf{D}, \mathbf{X}) \right)
$$
where the set of nuisance functions is now $\eta = (l, m, v)$.
\end{prop}

\begin{proof}[Proof of Score Simplification]
We need to evaluate the term $\frac{1}{p(\mathbf{D}|\mathbf{X})}\frac{\partial p(\mathbf{D}|\mathbf{X})}{\partial \mathbf{D}}$, which is the score of the log-likelihood, $\frac{\partial \log p(\mathbf{D}|\mathbf{X})}{\partial \mathbf{D}}$. Under Assumption 1, the conditional log-density is:
$$
\log p(\mathbf{D}|\mathbf{X}) = -\frac{1}{2}\log(2\pi v(\mathbf{X})) - \frac{(\mathbf{D}-m(\mathbf{X}))^2}{2v(\mathbf{X})}
$$
Taking the partial derivative with respect to $\mathbf{D}$ (treating $m(\mathbf{X})$ and $v(\mathbf{X})$ as constant with respect to $\mathbf{D}$):
$$
\frac{\partial \log p(\mathbf{D}|\mathbf{X})}{\partial \mathbf{D}} = - \frac{2(\mathbf{D}-m(\mathbf{X}))}{2v(\mathbf{X})} = - \frac{\mathbf{D}-m(\mathbf{X})}{v(\mathbf{X})}
$$
Substituting this result into the general score function from Proposition 1 yields the desired practical score:
$$
\psi(\dots) = \frac{\partial l}{\partial \mathbf{D}} - \theta - \left( - \frac{\mathbf{D}-m(\mathbf{X})}{v(\mathbf{X})} \right) (\mathbf{Y} - l(\mathbf{D}, \mathbf{X})) = \frac{\partial l}{\partial \mathbf{D}} - \theta + \frac{\mathbf{D}-m(\mathbf{X})}{v(\mathbf{X})} (\mathbf{Y} - l(\mathbf{D}, \mathbf{X}))
$$
\end{proof}

\begin{proof}[Proof of Orthogonality for the Practical Score]
The score in Proposition 2 is orthogonal to perturbations in $l$ by the same argument as in Proposition 1. We must also show it is orthogonal to perturbations in the new nuisance functions $m$ and $v$. The argument is as follows:

\begin{itemize}
    \item \textbf{Orthogonality w.r.t. $m$}: Let $m_t = m_0 + t[r_m]$ be a perturbed path for the true function $m_0$. The Gateaux derivative of the expected score is:
    $$
    \left. \frac{\partial}{\partial t} \mathbb{E}\left[ \psi(\mathbf{W}; \theta_0, l_0, m_t, v_0) \right] \right|_{t=0} = \mathbb{E}\left[ \left.\frac{\partial}{\partial t} \left( \frac{\mathbf{D} - m_t(\mathbf{X})}{v_0(\mathbf{X})} \right)\right|_{t=0} (\mathbf{Y} - l_0(\mathbf{D}, \mathbf{X})) \right] = \mathbb{E}\left[ \frac{-r_m(\mathbf{X})}{v_0(\mathbf{X})} (\mathbf{Y} - l_0) \right]
    $$
    By the Law of Iterated Expectations, conditioning on $\mathbf{X}$:
    $$
    \mathbb{E}_{\mathbf{X}}\left[ \frac{-r_m(\mathbf{X})}{v_0(\mathbf{X})} \mathbb{E}[\mathbf{Y} - l_0(\mathbf{D}, \mathbf{X}) | \mathbf{X}] \right]
    $$
    Since $\mathbb{E}[\mathbf{Y} - l_0(\mathbf{D}, \mathbf{X}) | \mathbf{X}] = \mathbb{E}[\mathbb{E}[\mathbf{Y} - l_0 | \mathbf{D}, \mathbf{X}] | \mathbf{X}] = \mathbb{E}[0 | \mathbf{X}] = 0$, the entire expression is zero.

    \item \textbf{Orthogonality w.r.t. $v$}: Similarly, for a path $v_t = v_0 + t[r_v]$, the Gateaux derivative is:
    $$
    \left. \frac{\partial}{\partial t} \mathbb{E}\left[ \psi(\mathbf{W}; \theta_0, l_0, m_0, v_t) \right] \right|_{t=0} = \mathbb{E}\left[ \left.\frac{\partial}{\partial t} \left( \frac{\mathbf{D} - m_0(\mathbf{X})}{v_t(\mathbf{X})} \right)\right|_{t=0} (\mathbf{Y} - l_0) \right]
    $$
    $$
    = \mathbb{E}\left[ -\frac{\mathbf{D} - m_0(\mathbf{X})}{v_0(\mathbf{X})^2}r_v(\mathbf{X}) (\mathbf{Y} - l_0) \right]
    $$
    Again, conditioning on $\mathbf{X}$:
    $$
    \mathbb{E}_{\mathbf{X}}\left[ -\frac{r_v(\mathbf{X})}{v_0(\mathbf{X})^2} \mathbb{E}[(\mathbf{D} - m_0(\mathbf{X}))(\mathbf{Y} - l_0(\mathbf{D}, \mathbf{X})) | \mathbf{X}] \right]
    $$
    The inner expectation is zero because $\mathbb{E}[(\mathbf{D} - m_0)(\mathbf{Y} - l_0) | \mathbf{X}] = \mathbb{E}_{\mathbf{D}|\mathbf{X}}[(\mathbf{D} - m_0) \mathbb{E}[\mathbf{Y}-l_0|\mathbf{D},\mathbf{X}] | \mathbf{X}] = \mathbb{E}_{\mathbf{D}|\mathbf{X}}[(\mathbf{D} - m_0) \cdot 0 | \mathbf{X}] = 0$. Thus, the score is orthogonal to perturbations in both $m$ and $v$.
\end{itemize}
\end{proof}

\section{Robustness of the Median Estimator for the APE}

\label{sec:appendix_ape_median}

In this appendix, we describe the justification for using a median-of-scores estimator and non-parametric bootstrap for inference, instead of the sample mean of the scores and its analytical variance. It is essential to make this choice in order ensure our causal estimates are robust against the estimation noise from the machine learning nuisance functions, which could cause score distributions to have fat tails and high skewness.

\subsection{The Problem: Unreliable Estimation with Heavy-Tailed Scores}

The practical score function for APE has a bias correction term that is inversely proportional to the conditional variance of treatment: $\hat{v}(\mathbf{X})$ (this treatment nuisance function is fitted with a machine learning algorithm). When $\hat{v}(\mathbf{X})$ is predicted to be small or close to zero for some observations, the bias correction term produces extreme outliers in the scores distribution $\hat{\psi}i$.

When these outliers are present, the sample mean is no longer a reliable way to estimate the central tendency of the distribution. In Figure \ref{fig:ape_score_dist}, we demonstrate this problem for the treatment variable \seqsplit{amihud\_illiquidity\_trend\_z\_scaled\_std}. The distribution is shown to be skewed and heavy-tailed. There are a few extreme positive outliers that almost cancel the statistical weight of the much more numerous scores with small negative values. Thus, we get a sample mean ($0.0010$) that is deceivingly close to zero. A naive interpretation of the mean would have resulted in a mischaracterization of a null causal effect.

\begin{figure}[htbp]
\centering
\includegraphics[width=0.9\textwidth]{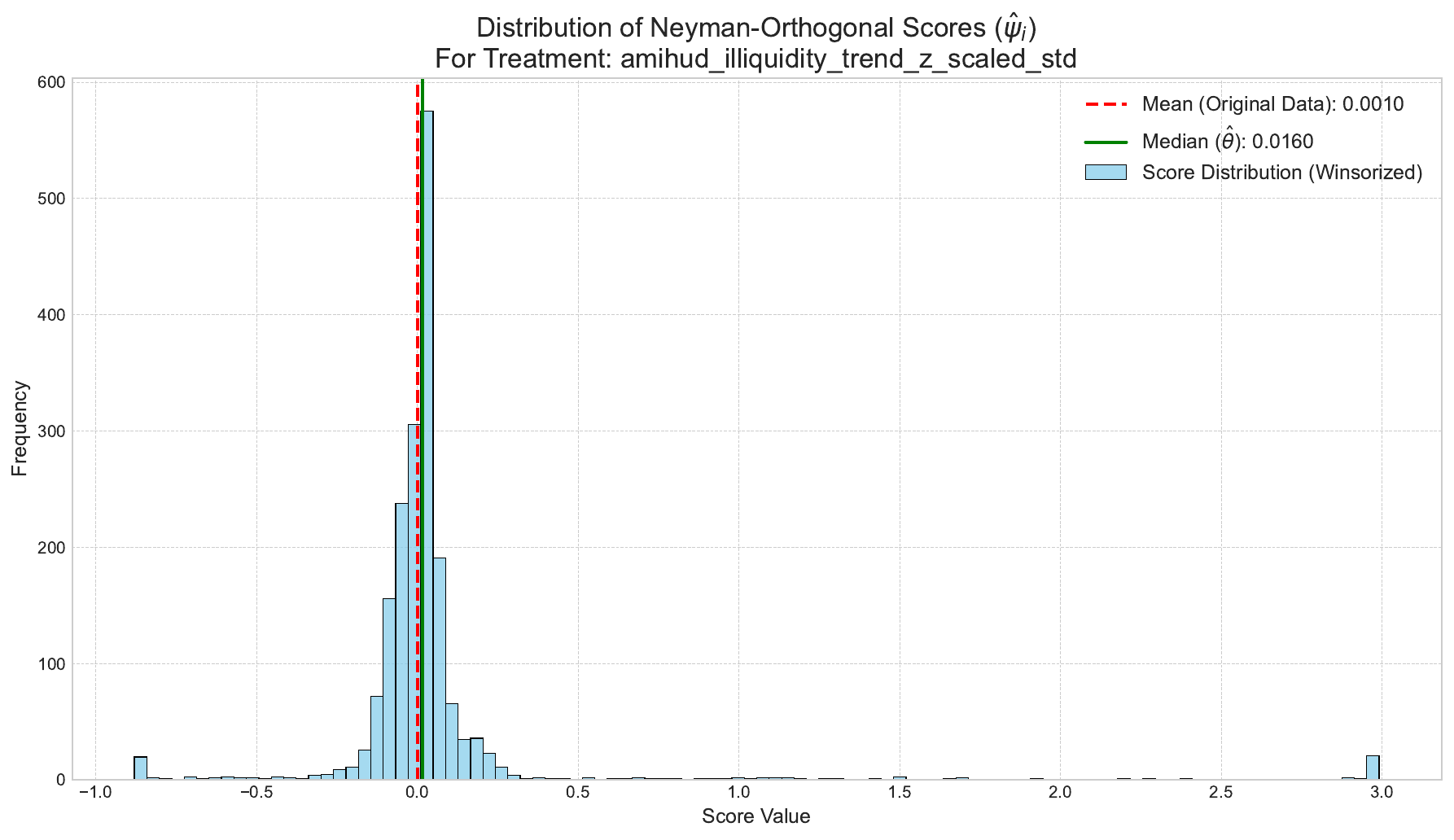} 
\caption{Distribution of Neyman-Orthogonal Scores for the APE}
\label{fig:ape_score_dist}
\justify
\small{\textit{Notes:} This figure shows the histogram of the calculated Neyman-orthogonal scores ($\hat{\psi}_i$) for the DML-APE model where the treatment variable is \seqsplit{amihud\_illiquidity\_trend\_z\_scaled\_std}. The distribution is heavy-tailed and skewed. The red dashed line indicates the sample mean of the original scores (0.0010). The green solid line marks the sample median (0.0160), which is used as the robust point estimate for the APE ($\hat{\theta}$). This plot demonstrates how the sample mean can be a misleading measure of central tendency; in this case, it is driven toward zero by the cancelling effects of outliers, while the median robustly captures the positive shift in the core of the distribution. For visualization, the scores have been winsorized at the 1st and 99th percentiles.}
\end{figure}

\subsection{The Solution: Median Estimator \& Bootstrap Inference}

Instead of calculating the mean, we can use the sample median as a robust point estimator. The median is the 50th percentile, and its actual value only depends on all of the scores at the center of the distribution, immune to the outliers are in the tails. Therefore, by using the median of the $\hat{\psi}_i$ scores as a point estimate for $\theta_0$, we have an estimator that is robust to noisy outputs from the nuisance models.

Since analytical formulas for the standard error are also not suitable for heavy-tailed distribution, we use a non-parametric bootstrap to estimate the standard-error for the median. We resample our calculated scores repeatedly and compute the median from each bootstrap sample, producing an empirical sampling distribution for our estimator. This bootstrapping allows us to produce a more credible standard-error.

\subsection{Significance}

This procedure is beneficial for credible inference. As shown in Figure \ref{fig:ape_score_dist}, if we calculated the sample mean, we would get a point estimate around zero and would erroneously concluded a null causal effect (Type II error). The sample median, on the other hand, identifies the positive and statically significant signal. The median/bootstrap median framework helps defend our analysis against the influence of outliers, enhancing the ability to identify credible economic conclusions, and systematically minimize the likelihood of erroneous null findings.

\section*{Online Appendix}
A supplementary online appendix, containing the complete results of the Double / Debiased Machine Learning (DML) estimation and the full sensitivity analysis for all variables tested, is available at the following persistent URL: \href{https://github.com/jackraorpl/market-trough-prediction-appendix/blob/89aaba27f707f9b223f0afbdc9f951d0f8155db2/ONLINE_APPENDIX.md}{github.com/jackraorpl/market-trough-prediction-appendix}.

\clearpage

\end{document}